\documentclass[journal]{IEEEtran}
\usepackage{amsmath}
\usepackage{amssymb}
\usepackage{amsfonts}
\usepackage{mathtools,amsthm}
\usepackage{comment}
\usepackage{graphicx}
\usepackage{caption}
\usepackage{subcaption}
\usepackage{hyperref}
\usepackage{xcolor}
\usepackage{bbold}

\theoremstyle{definition}

\newtheorem{proposition}{Proposition}
\usepackage[ruled,linesnumbered]{algorithm2e}
\usepackage{enumitem}
\usepackage{multirow}
\usepackage{diagbox}

\begin{document}

\title{A Reliable Vertical Federated Learning Framework for Traffic State Estimation with Data Selection and Incentive Mechanisms}

\markboth{Journal of \LaTeX\ Class Files,~Vol.~14, No.~8, August~2015}%
{Shell \MakeLowercase{\textit{et al.}}: Bare Demo of IEEEtran.cls for IEEE Journals}

\author{Zijun~Zhan,
        Yaxian~Dong,
        Daniel~Mawunyo~Doe,
        Yuqing~Hu,
        Shuai~Li,
        Shaohua Cao,
        and~Zhu~Han
        \thanks{Manuscript received ...}
        \thanks{Zijun Zhan is with the Department of Electrical and Computer Engineering, University of Houston, 4800 Calhoun Rd, Houston, TX 77004. E-mail: zzhan@uh.edu}
         \thanks{Yaxian Dong and Yuqing Hu are with the Department of Architectural Engineering, The Pennsylvania State University, University Park, PA 16802, USA (E-mail: yzd5221@psu.edu and yfh5204@psu.edu)}
	\thanks{Daniel Mawunyo Doe is with the Department of Electrical and Computer Engineering, Prairie View A\&M University, 100 University Dr, Prairie View, TX 77446. Email: dmdoe@pvamu.edu}
         \thanks{Shuai Li is with the Department of Civil \& Coastal Engineering, University of Florida, Gainesville, FL 32611. E-mail: shuai.li@ufl.edu}
         \thanks{Shaohua Cao is with the Qingdao Institute of Software, College of Computer Science and Technology, China University of Petroleum (East China), Qingdao 266580, China. E-mail:shaohuacao@upc.edu.cn}
         \thanks{Zhu Han is with the Department of Electrical and Computer Engineering, University of Houston, 4800 Calhoun Rd, Houston, TX 77004, and also with the Department of Computer Science and Engineering, Kyung Hee University, Seoul, South Korea, 446-701. E-mail: hanzhu22@gmail.com}
         \thanks{This research was funded by the U.S. National Science Foundation (NSF) via Grants 2222730, 2222670, and 2222810. In addition, this work is partially supported by NSF CNS-2107216, CNS-2128368, CMMI-2222810, ECCS-2302469, US Department of Transportation, Toyota, Amazon and JST ASPIRE JPMJAP2326.}
}

\maketitle

\begin{abstract}
Vertical Federated Learning (VFL)-based Traffic State Estimation (TSE) offers a promising approach for integrating vertically distributed traffic data from municipal authorities (MA) and mobility providers (MP) while safeguarding privacy. However, given the variations in MPs' data collection capabilities and the potential for MPs to underperform in data provision, we propose a reliable VFL-based TSE framework that ensures model reliability during training and operation. The proposed framework comprises two components: data provider selection and incentive mechanism design. Data provider selection is conducted in three stages to identify the most qualified MPs for VFL model training with the MA. First, the MA partitions the transportation network into road segments. Then, a mutual information (MI) model is trained for each segment to capture the relationship between data and labels. Finally, using a sampling strategy and the MI model, the MA assesses each MP's competence in data provision and selects the most qualified MP for each segment. For the incentive mechanism design, given the MA can leverage the MI mode to inspect the data quality of MP, we formulate the interaction between MA and MP as a supervision game model. Upon this, we devise a penalty-based incentive mechanism to inhibit the lazy probability of MP, thereby guaranteeing the utility of MA. Numerical simulation on real-world datasets showcased that our proposed framework augments the traffic flow and density prediction accuracy by 11.23\% and 23.15\% and elevates the utility of MA by 130$\sim$400\$ compared to the benchmark.
\end{abstract}

\begin{IEEEkeywords}
Vertical federated learning (VFL), traffic state estimation (TSE), mutual information, data provider selection, incentive mechanism design
\end{IEEEkeywords}

\IEEEpeerreviewmaketitle

\section{Introduction}

\IEEEPARstart{T}{raffic} state estimation (TSE) is a crucial component within the domain of intelligent transportation systems. Accurate and timely traffic information enables more efficient route planning, alleviates congestion, and improves overall road safety. To achieve higher prediction accuracy, a variety of methods based on deep learning and traffic flow theory have been extensively researched \cite{pan2023ising,yang2023traffic,xu2024pigat}. However, the training of TSE models often requires traffic data provided by different entities, which introduces potential data privacy concerns. Specifically, traffic data can be classified into two categories: Eulerian observations and Lagrangian observations \cite{wang2024privacy}. The former is collected from roadside sensors managed by municipal authorities (MAs), while the latter is derived from vehicle trajectories provided by mobility providers (MPs) with extensive fleets. Given that the data from MPs can reveal sensitive information, such as service areas, demand patterns, fleet sizes, and vehicle scheduling algorithm parameters \cite{he2019optimal}, these providers may be hesitant to share their data directly for TSE model training.

To cope with similar issues, the authors in \cite{ye2020federated} endeavored to integrate federated learning into vehicular networks, thereby training the neural network model in a privacy-safeguarding manner. Notably, the prior work is generally focused on horizontal federated learning (HFL); however, the data supplied by MA and MPs is vertically partitioned since they will provide the data on different road segments, posing the demand for vertical federated learning (VFL) \cite{liu2024vertical}. Several pioneering studies have already demonstrated the effectiveness and advantages of using VFL for traffic state estimation \cite{wang2024privacy, chougule2023novel, errounda2022mobility}. For example, the authors in \cite{wang2024privacy} applied VFL to integrate data from MAs and MPs for real-time traffic state estimation, incorporating traffic models into VFL to address the issue of limited ground-truth data. Similarly, the authors in \cite{chougule2023novel} proposed an asynchronous VFL framework to reduce vehicle idle times at red lights by integrating data from different MPs. Likewise, the authors in \cite{errounda2022mobility} combined data from various MPs and introduced a mobility-oriented vertical federated forecasting framework to capture spatiotemporal relationships between road segments, thereby enhancing mobility predictions across a joint location domain.

These studies have made initial efforts to fuse traffic data from MA and MPs to refine spatial road network representations and enhance traffic prediction performance. Nevertheless, the current application of VFL in vehicular networks operates under the assumption that there is no overlap in data provision between MPs, implying that each road segment's data is provided exclusively by a single entity. This assumption fails to account for the reality of overlapping data provision, where multiple MPs may collect and provide data on the same road segments. Furthermore, differentiated data collection techniques and data imputation methods may render varying data quality across MPs within the same service area, which leads to the first research question:

\emph{Q1: How can high-quality data providers be effectively screened for reliable VFL-based TSE model training?}

Additionally, once the TSE model is well-trained among MAs and selected MPs, there still is a risk that MPs may become sloth, resorting to supplying randomly generated or historical data rather than real-time data to augment their utility. For this reason, the second research question is elicited, i.e.,

\emph{Q2: How can data providers be spurred to provide high-quality data for reliable VFL-based TSE model running?}

\begin{figure*}[!t]
    \begin{center}
        \includegraphics[width=\linewidth]{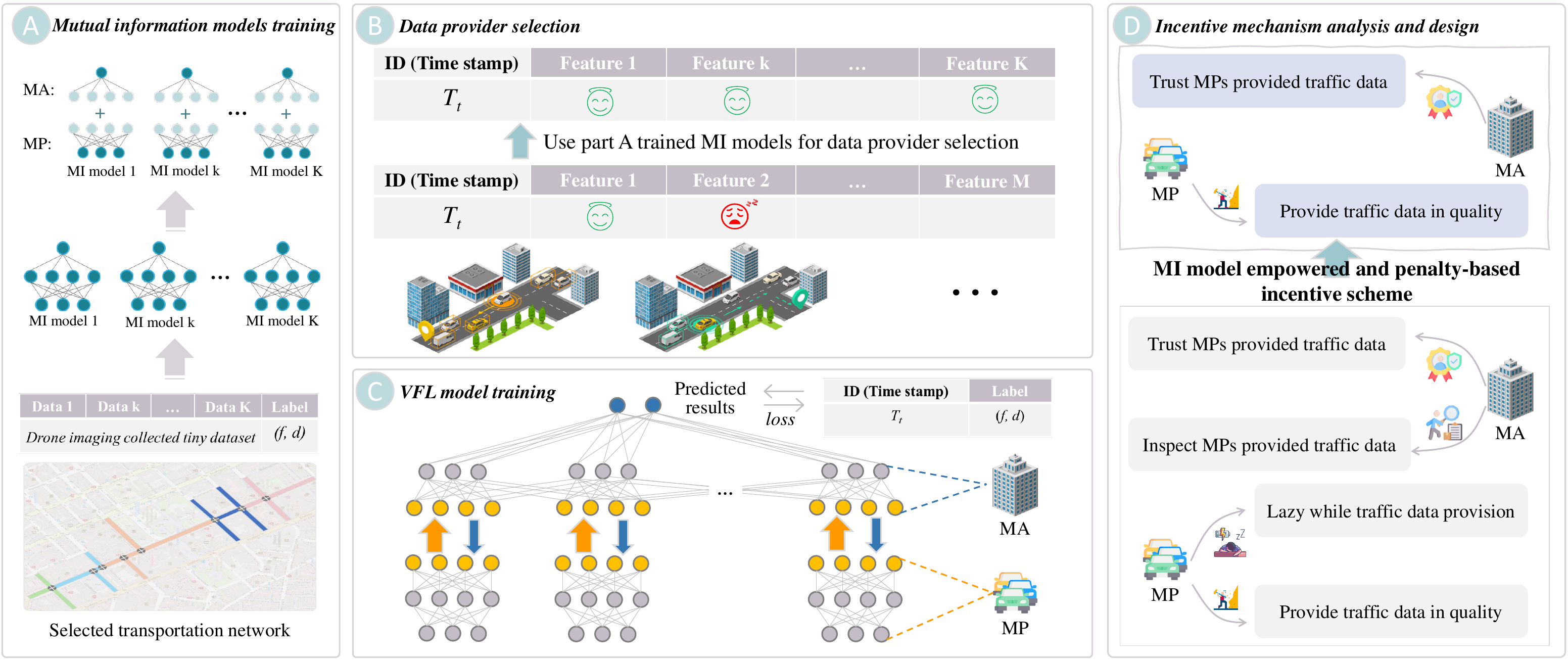}
        \caption{Framework overview of the proposed VFL-based TSE with data provider selection and the incentive mechanism design.}
        \label{fig1}
    \end{center}
\end{figure*}

Regarding Q1, this challenge can be viewed as a specialized feature selection problem within the VFL framework, which has been extensively studied in prior research \cite{li2023fedsdg,jiang2022vf,zhang2022secure,feng2022vertical, castiglia2023less}. However, recent feature selection methods in VFL have predominantly focused on classification tasks and may be constrained by the dimensionality of the training data. Notably, TSE models typically leverage graph neural networks (GNNs) for continuous value prediction. Furthermore, GNN-based TSE models often operate in a high-dimensional space due to the extensive number of roads in a transportation network and the multiple historical time steps required for forecasting. Moreover, existing feature selection approaches adopt a global perspective, aiming to identify indispensable features from the entire feature set. In contrast, the selection of MPs for the TSE model is from a local view, that is, opting for the top-performing MP among those capable of supplying data for the same road segment.

Upon this, we propose the integration of a mutual information (MI) model within the VFL framework to screen MPs, as illustrated in parts A and B of Fig. \ref{fig1}. Specifically, as shown in Part A of Fig. \ref{fig1}, we begin with partitioning the transportation network into multiple road segments, with each segment represented by a unique color. Then, the MA conducts drone imaging to collect a small dataset covering the entire transportation network and utilize it to train the MI model for each road segment, capturing the relationship between the data and the corresponding road segments. Subsequently, the MA vertically splits the trained MI models, distributing the bottom MI model to the respective data providers while retaining the top MI model locally to maintain the property of VFL. Finally, as depicted in Part B, the MA randomly samples a set of IDs and employs the MI models to evaluate the MPs, selecting the top-performing MP for each road segment to participate in VFL-based TSE model training.

Notably, once the participating MPs are selected, the VFL-based TSE model proceeds with training as illustrated in Part C of Fig. \ref{fig1}. Specifically, each selected MPs leverage their local data to train a bottom sub-model and upload the outputs of the model (neurons highlighted in yellow), referred to as intermediate results, to the MA for model aggregation. The MA then processes the aggregated intermediate results by training the top sub-model to generate the predicted traffic states (neurons highlighted in blue). Following this, the MA computes the overall model loss based on the true labels and the predicted traffic states and uses this loss to derive gradients with respect to both the top sub-model and the intermediate results provided by each MP. Finally, the MA and MPs update the parameters of the top sub-model and the individual bottom sub-models based on the calculated gradients, thereby completing the training process.

Regarding Q2, numerous prior studies \cite{arisdakessian2023coalitional, liu2024chiron, li2024rate, pan2024towards, zhao2024context, fu2023incentive} have explored incentive methods to mitigate lazy or selfish behavior among federated learning (FL) clients within the context of horizontal federated learning (HFL). However, limited research has focused on incentivizing FL clients in the vertical federated learning (VFL) setting. To fill this gap, we propose an MI model-driven and penalty-based incentive mechanism to motivate MPs to supply high-quality traffic data during TSE, as depicted in Part D of Fig. \ref{fig1}. Concretely, MPs will have two options for real-time TSE: either be lazy while traffic data provision or provide high-quality traffic data. Given that the MI models trained in Part A can be used to verify the traffic data supplied by MPs, the MA will also have two options: either trust the data provided by MPs or inspect it. We model the interactions between the MA and MPs in TSE as a supervision game and design a penalty-based incentive mechanism. This mechanism promotes an equilibrium through game evolution, ensuring that the MA trusts the quality of the data provided by MPs, while MPs are incentivized to deliver high-quality traffic data for TSE. By doing so, we can guarantee the fidelity of the predicted traffic states when MA and MPs collectively run the VFL-based TSE model trained in Part C of Fig. \ref{fig1}.

In summary, the contributions of this paper are presented below:
\begin{enumerate}
    \item[1.] We propose a reliable VFL-based traffic state estimation (TSE) framework throughout the VFL model training and running processes. By leveraging the MI model and penalty-based incentive mechanism, our framework not only ensures the selection of competent MPs for VFL-based TSE model training but also guarantees that MPs consistently supply high-quality data for traffic state prediction.
    \item[2.] An MI model-driven data provider selection scheme is introduced for reliable VFL-based TSE model training. Given the complexity of GNN-based TSE models, which involve continuous value prediction in a high-dimensional space, we employ the MI model to approximate the capabilities of MPs in traffic data provision. By doing so, the most competent MPs are screened for reliable VFL-based TSE model training.
    \item[3.] We develop an MI model-driven and penalty-based incentive mechanism for reliable traffic state prediction. Utilizing the MI model, we formulate the interaction strategy between MAs and MPs during traffic state prediction as a supervision game model. Furthermore, we theoretically analyze the feasibility of our proposed incentive mechanism.
    \item[4.] Extensive numerical simulation experiments were conducted to evaluate the performance of the proposed framework. The results demonstrate that the data provider selection scheme improves the accuracy of traffic flow and density predictions by $11.23\%$ and $21.15\%$ compared to benchmark models. Moreover, under various simulation settings, the penalty-based incentive mechanism augments the utility of MA by around $130\sim 400 \$$ compared to the benchmark.
\end{enumerate}

The remainder of this paper is organized as follows: Section \ref{sec:2} provides a review of the relevant literature. In Section \ref{sec:3}, we introduce the VFL-based TSE model and the utility model of MP and MA. Sections \ref{sec:4} and \ref{sec:5} address Q1 and Q2 through mathematical analysis and present the corresponding solutions and algorithms. The performance of the proposed MI-based data provider selection scheme and penalty-driven incentive mechanism is evaluated in Section \ref{sec:6}. Finally, the paper concludes in Section \ref{sec:7}.

\section{Related Works} \label{sec:2}
Given two research questions, Q1 and Q2, proposed in this paper are in the scope of feature selection in VFL and incentive mechanism design for selfish clients in FL, we will review the state-of-art literature in these two scopes in Sections \ref{sec:2.1} and \ref{sec:2.2}, as well as elaborate the key difference between our work and prior research.

\subsection{Feature Selection in VFL} \label{sec:2.1}
Feature selection is a critical problem in the domain of VFL since certain features might be trivial and even deteriorate the training performance of the VFL model. To this end, numerous works have been done to optimize the features used for VFL model training. For instance, the authors in \cite{zhang2022secure} proposed a Gini-impurity-based feature selection framework for VFL, in which the averaged Gini score of each feature is the criteria for feature selection. Advancing the prior study, the authors in \cite{li2023fedsdg} began with initializing the importance of each feature via Gini-impurity and then devised a trainable stochastic dual-gate model to determine the selection probability of features. It is worth noting that the dual-gate model is embedded in the VFL model, which means the features will be screened during the VFL model training. Analogously, the authors in \cite{feng2022vertical} and \cite{castiglia2023less} also proposed the embedded approach to tackle the feature selection problem in VFL, the primary difference between these methods is in the design of the loss function for the VFL model. Diverging from previous works, the authors in \cite{jiang2022vf} leveraged the KNN-based MI estimator and group-based sampling strategy to approximate the score of each feature, and then select certain most important features out of all features.

In spite of the aforementioned studies having made substantial contributions to feature selection in VFL, it is challenging to adopt these methods for MP selection in the context of VFL-based TSE. Concretely, previous feature selection schemes were primarily focused on classification tasks, and the VFL model has limited dimensions. Additionally, the MP selection problem in this paper aims to screen the most competent MP for each road segment, which is a little bit different from the traditional feature selection problem defined in the VFL. Therefore, we devise an MI model-based feature selection scheme, in this paper, to fill these gaps.

\subsection{Incentive Mechanism Design for FL} \label{sec:2.2}
In light of selfish or lazy clients who will compromise the fairness of the FL system, and even deteriorate the training performance of the FL model, various incentive mechanisms have been introduced recently. The authors in \cite{arisdakessian2023coalitional} assess whether FL clients are lazy or not based on the FL client's subjective and objective trust score. The subjective score is calculated from other FL clients' recommendations on the behavior of the FL client and the objective score is obtained as per the resource utilization of FL clients during the Model training. In light of the incomplete information for the FL server during the incentive design, the authors in \cite{liu2024chiron} introduced the experience-driven DRL and PCA method to screen the lazy FL clients, in which the PCA is used to compress the model uploaded by each client as the state of DRL model. Regarding the GNN-based FL model, the authors in \cite{pan2024towards} leveraged a gradient-based shapely value approach to measure the contribution of each FL client, so as to incentivize the FL client based on their contributions. The authors in \cite{li2024rate} first formulated the interaction between FL server and clients as a Stackelberg game and then proposed a reputation-based incentive mechanism to incentivize FL clients to continuously supply data in quality. Similarly, the authors in \cite{fu2023incentive} formulated the interaction between the FL server and the FL client as a supervision game and devised a reputation and reward-based incentive mechanism to inhibit the lazy probability of FL clients. Notably, the authors in \cite{zhao2024context} unleashed the power of blockchain to resolve the free-rider issue in the context of FL. Specifically, the authors designed a contribution maximization consensus mechanism to filter out lazy clients in the FL.

The common ground of aforesaid literature is that the incentive mechanism is designed in the context of HFL, which might not work in the VFL setting due to the wholly distinct training procedures of HFL and VFL. To this end, we proposed a penalty-based incentive mechanism, with the aid of the MI model, to fill this knowledge gap.

\section{System Model} \label{sec:3}
The system model of our proposed framework is presented in this section, in which we first introduce the VFL-based TSE model in Section \ref{sec:3.1} and then delineate the utility model of MP and MA in the context of real-time TSE in Sections \ref{sec:3.2} and \ref{sec:3.3}, respectively.

\subsection{VFL-Based TSE Model} \label{sec:3.1}
As per the system framework presented in Fig. \ref{fig1}, we assume the VFL-based TSE model occurs in the scenario of an urban transportation network which is described as $\mathcal{G}=(\mathcal{V}, \mathcal{E})$. Here, $\mathcal{V}$ denotes all the road intersections and $\mathcal{E}$ indicates all the roads, in which $|\mathcal{V}| = 8$ and $|\mathcal{E}| = 21$ to line with the transportation network depicted in the Part A of Fig. \ref{fig1}. Upon the selected urban transportation network, $M$ data providers, including one MA and multiple MPs, will provide the traffic data on $\mathcal{E}$ collectively for VFL-based TSE model training. For clarity, the MA will be deemed as a special type of MP in our model, in which the data provided by MA is more trustworthy than MPs of the same road segment. Notably, MA will also serve as the label owner to orchestrate the VFL model training. Analogous to \cite{wang2024privacy}, we assume the label possessed by MA for model training is acquired via drone imaging.

In this paper, $M$ MPs are denoted as $\mathcal{M}=\{M_m | m \in [M]\}$, where $[M]=\{1,\cdots,M\}$ is the index of $M$ MPs. In light of business overlap that might exist between MPs, we assume the roads set $\mathcal{E}$ can be partitioned into $K$ road segments, denoted as $\mathcal{K}=\{K_k | k \in [K]\}$, and the traffic data on each road segment $K_k$ can be provided by $N$ MPs. We utilize the set $\mathcal{M}_k=\{M_{k,n} | n \in [N]\}$ to denote the $N$ MPs who monitor the same road segment $K_k$. In this way, prior to the VFL model training, MA needs to screen the most representative $M_{k,n^{*}}$ from $\mathcal{M}_k$ for each road segment $K_k$, i.e., data provider selection problem, which will be resolved in Section \ref{sec:4}.

Regarding the dataset utilized for VFL model training, analogous to \cite{wang2024privacy,errounda2022mobility}, we consider the ID sequence of training data and label is composed of a series of $T$ time slots, which is denoted by $\mathcal{T}=\{T_t|t\in[T]\}$, in which, the sampling interval of $\mathcal{T}$ is denoted as $\Delta T= {T_{t+1}-T_{t}}$. By doing so, the training data possessed by $M_{k,n}$ is denoted as $\mathcal{X}_{k,n}=\{x_{k,n}^{T_t}|t\in[T], x_{k,n}^{T_t} \in \mathbb{R}^{\tau_i \times |K_k| \times s}\}$ and the label possessed by MA is denoted as $\mathcal{Y}=\{y^{T_t}|t\in[T], y^{T_t} \in \mathbb{R}^{\tau_o \times |\mathcal{E}| \times s}\}$. Here, $\tau_i$ represents the length of the historical time steps as the input of the VFL-based TSE model, $\tau_o$ is the length of the prediction time steps of the VFL-based TSE model, and $s$ is the number of traffic states to be forecast. In this paper, we set the TSE model to predict the traffic state of flow and density.

Subsequent to the definition of MPs and dataset, we utilize $\mathbf{A}_{K \times N}$ to represent the binary selection matrix of MA, i.e., $a_{k,n} \in \{0, 1\}$, where $a_{k,n} = 1$ means the $M_{k,n}$ is selected by MA to supply the traffic data for road segment $K_k$. Subsequently, MA and $K$ selected MPs will collectively train the VFL model. Here, we utilize $\mathbf{D}_{K \times N}$, $\Phi_{K \times N}$, and  $\Theta_{K \times N}$ to represent the matrix form of the data, model, and model parameters of $K \times N$ MPs, respectively. In this way, the data, model, and model parameters of $K$ MPs can be identified via $\mathbf{A} \circ \mathbf{D}$, $\mathbf{A} \circ \Phi$, and $\mathbf{A} \circ \Theta$, in which $\circ$ is the Hadamard Product. By doing so, the VFL model training process can be described as follows.

Firstly, each selected data provider $M_{k,\tilde{n}}$ will utilize its local sub-model $\phi_{k,\tilde{n}}(\cdot)$, model parameter $\theta_{k,\tilde{n}}$, and data $\mathcal{X}_{k,\tilde{n}}$ to generate the intermediate output $z_{k,\tilde{n}}$, which is denoted as:
\begin{equation} \label{eq:1}
    z_{k,\tilde{n}} = \phi_{k,\tilde{n}}(\theta_{k,\tilde{n}};\mathcal{X}_{k,\tilde{n}}).
\end{equation}

Secondly, the MA will collect and leverage all the intermediate outputs $\mathcal{Z}=\{z_{1,\tilde{n}},\cdots,z_{K,\tilde{n}}\}$ to predict the traffic state $\tilde{\mathcal{Y}}$, which is specified as:
\begin{equation} \label{eq:2}
    \tilde{\mathcal{Y}} = \phi_{0}(\theta_{0};\mathcal{Z}),
\end{equation}
where $\phi_0$ and $\theta_0$ are the model and model parameters of the MA.

Thirdly, as per the label and predicted results, MA will calculate the loss by the specified loss function $\ell (\cdot)$,
\begin{equation} \label{eq:3}
    \mathcal{L}(\mathbf{A}, \mathbf{D}, \Theta, \theta_0, \mathcal{Y})=\ell(\tilde{\mathcal{Y}},\mathcal{Y}).
\end{equation}

Eventually, the VFL-based TSE model training leverages gradient descent algorithms to iteratively optimize the loss function depicted in (\ref{eq:3}).

\subsection{Utility Model of MP} \label{sec:3.2}
In the context of VFL-based TSE, the MP will experience traffic data collection, traffic data processing, sub-model running, and intermediate results uploading, which are four major components of the cost of MPs. Accordingly, the MP's payment is the reward issued by the MA, and the payment should be greater than the cost to ensure incentive rationality.

\textbf{Data collection}: Given that traffic data is collected via a fleet of vehicles, we assume the cost of data collection mainly stems from the energy cost of uploading the data from vehicles to MP. With reference to \cite{kazmi2021novel}, we formulate the uplink data rate of the communication model as
\begin{equation} \label{eq:4}
    {R(t)} = B(t){\log _2}\left(1 + \frac{{{p}{g(t)}}}{{{N_0}}}\right),
\end{equation}
where $B(t)$ is the allocated bandwidth between vehicle and MP at time $t$, $p$ is the transmission power of the vehicle, and $g(t)$ is the channel gain between vehicle and MP at time $t$. Here, $g(t)$ is defined as $ = 128.1 + 37.5 log10 (d(t))$, in which $d(t)$ is the distance between the vehicle and MP at time $t$. Notably, similar to \cite{ning2020intelligent}, we assume the communication between vehicles and MP using the orthogonal frequency-division multiplexing technique, i.e., no interference from other vehicles during the traffic data uploading. Therefore, $N_0$ indicates the Additive White Gaussian Noise.

Since the traffic data uploaded from vehicles to MP are in the same format, we utilize $D$ to denote the data size of one traffic data. Furthermore, as the dataset description in Section \ref{sec:3.1} mentioned the data sampling interval is $\Delta T$, the traffic data collected during this interval should all uploaded to MP for data processing and generate the data for TSE model training. We assume the data collection interval of vehicles is defined as $\Delta t$. Consequently, the cost of data collection for a piece of TSE model required data, $x_{k,n}^{T_t}$, is formulated as
\begin{equation} \label{eq:5}
    E_t = \frac{\Delta T \times p \times D}{\Delta t \times R(t)}.
\end{equation}

\textbf{Data processing}: According to the cost analysis of data collection and refer to \cite{fu2023incentive}, we design the cost of data processing for a piece of TSE model required data as
\begin{equation} \label{eq:6}
    E_p = \frac{\Delta T}{\Delta t} \eta c D f_c^2,
\end{equation}
where $\eta$ is the effective capacitance coefficient of the chip of MP, $c$ is the number of CPU cycles required by MP for processing a byte of traffic data, and $f_c$ is the CPU cycle frequency of MP.

\textbf{Sub-model running}: Due to the running of the sub-model will utilize CPU and GPU simultaneously and the cost will vary along the size of the deep learning model possessed by MP, causing difficulty in modeling the cost of sub-model running. For this reason, we utilize the codecarbon function \footnote{https://pypi.org/project/codecarbon/} $f(\cdot)$ to inscribe the cost of the sub-model running, i.e.,
\begin{equation} \label{eq:7}
    E_r = f(\phi_{k,\bar{n}}, x_{k,n}^{T_t}).
\end{equation}

\textbf{Intermediate results uploading}: Upon the intermediate results is the final layer of the sub-model, which only accounts for several bytes, thus the cost is negligible. Therefore, we formulate the utility model of the MP for one settlement cycle in the period of the real-time TSE as
\begin{equation} \label{eq:8}
    U_{MP} = W - \tau_e \mathbb{T} (E_t + E_p + E_r).
\end{equation}
Here, $W$ denotes the reward issued by MA after one settlement cycle, $\tau_e$ is the energy and currency conversion coefficient \cite{kim2017dual}, and $\mathbb{T}$ means the number of TSE model's input data in the period of one settlement cycle.

\subsection{Utility Model of MA} \label{sec:3.3}
Distinct from the cost analysis of MP, the only expenditure of MA is the reward issued to MP. Regarding the profit of the MA, inspired from \cite{lim2020hierarchical}, we consider that the profit of the MA is proportional to the model performance. Consequently, we define the utility model of the MA in one settlement cycle as
\begin{equation} \label{eq:9}
    {U_{MA}} = \frac{1}{K}{\tau _a} \mathbb{I} (\tilde x - x) - W,
\end{equation}
in which the indicator $\mathbb{I}$ is defined as
\begin{equation} \label{eq:10}
    \mathbb{I}  = \left\{ \begin{array}{l}
    1 ,\quad \mbox{if } x < \tilde{x},\\
    0,\quad \mbox{otherwise}.
    \end{array} \right.
\end{equation}

Here, $1/K$ indicates the contribution of each MP is identical to the TSE model since missing the traffic data of any MP will render incomplete spatiotemporal topology of the transportation network, thereby significantly degrading the performance of the TSE model. $\tau_a$ is the model performance and currency conversion coefficient, $\tilde{x}$ indicates the threshold performance that a viable TSE model should meet for usability, and $x$ is the averaged TSE model performance for one settlement cycle.

\section{MI-Based Data Provider Selection} \label{sec:4}
In this section, we begin with mathematically formulating the data provider selection problem in Section \ref{sec:4.1}. Subsequently, the neural network-based MI value approximation method is delineated in Section \ref{sec:4.2}. Eventually, in Section \ref{sec:4.3}, we illustrate the resolution of the problem formulated in Section \ref{sec:4.1}.  
\subsection{Problem Formulation} \label{sec:4.1}
Upon the introduction of the VFL-based TSE model in Section \ref{sec:3.1}, the data provider selection problem that we need to resolve can be written as follows:
\begin{equation} \label{eq:11}
    \begin{aligned}
    \mathcal{P}1: ~   \min_{\mathbf{A},\Theta, \theta_0} ~ &\mathcal{L}(\mathbf{A}, \mathbf{D}, \Theta, \theta_0, \mathcal{Y}) \\
s.t. &\sum_{n=1}^{N} a_{k, n} = 1, \forall~k \in [K], \\
&a_{k,n} \in \{0,1\}, \forall~k \in [K], \forall~n \in [N].
    \end{aligned}
\end{equation}
Here, the two constraints indicate that for each road segment $K_k$ only one traffic data provider will be selected.

Regarding the problem $\mathcal{P}1$ stated in (\ref{eq:11}), we will decompose it into two sub-problems, which are MPs selection problem $\mathcal{P}2$
\begin{equation} \label{eq:12}
        \begin{aligned}
    \mathcal{P}2: ~   \max_{\mathbf{A}} ~ & \sum_{k=1}^{K}{\sum_{n=1}^{N}{a_{k,n} \times Q(\mathcal{X}_{k,n})}} \\
s.t. &\sum_{n=1}^{N} a_{k, n} = 1, \forall~k \in [K], \\
&a_{k,n} \in \{0,1\}, \forall~k \in [K], \forall~n \in [N],
    \end{aligned}
\end{equation}
and loss minimization problem $\mathcal{P}3$
\begin{equation}  \label{eq:13}
    \begin{aligned}
    \mathcal{P}3: ~   \min_{\Psi} ~ &\mathcal{L}(\Psi,\mathcal{D}).
    \end{aligned}
\end{equation}
The physical meaning of $\mathcal{P}2$ is to identify the most representative MP among each road segment, in which $Q(\mathcal{X}_{k,n})$ is the quality of traffic data possessed by $M_{k,n}$.

Upon the completion of MPs selection, $\mathcal{P}1$ can be simplified as $\mathcal{P}3$, where $\Psi = \{\theta_0\} \cup \{\theta_{k,n} | a_{k,n} = 1\}$ indicates all the model parameters possessed by MA and $K$ selected MPs. Analogously, $\mathcal{D} = \{\mathcal{Y}\} \cup \{\mathcal{X}_{k,n} | a_{k,n} = 1\}$ is the label and traffic data of MA and MPs.

\subsection{Neural Network-Based MI Model} \label{sec:4.2}
\begin{figure}[!t]
    \begin{center}
        \includegraphics[width=\linewidth]{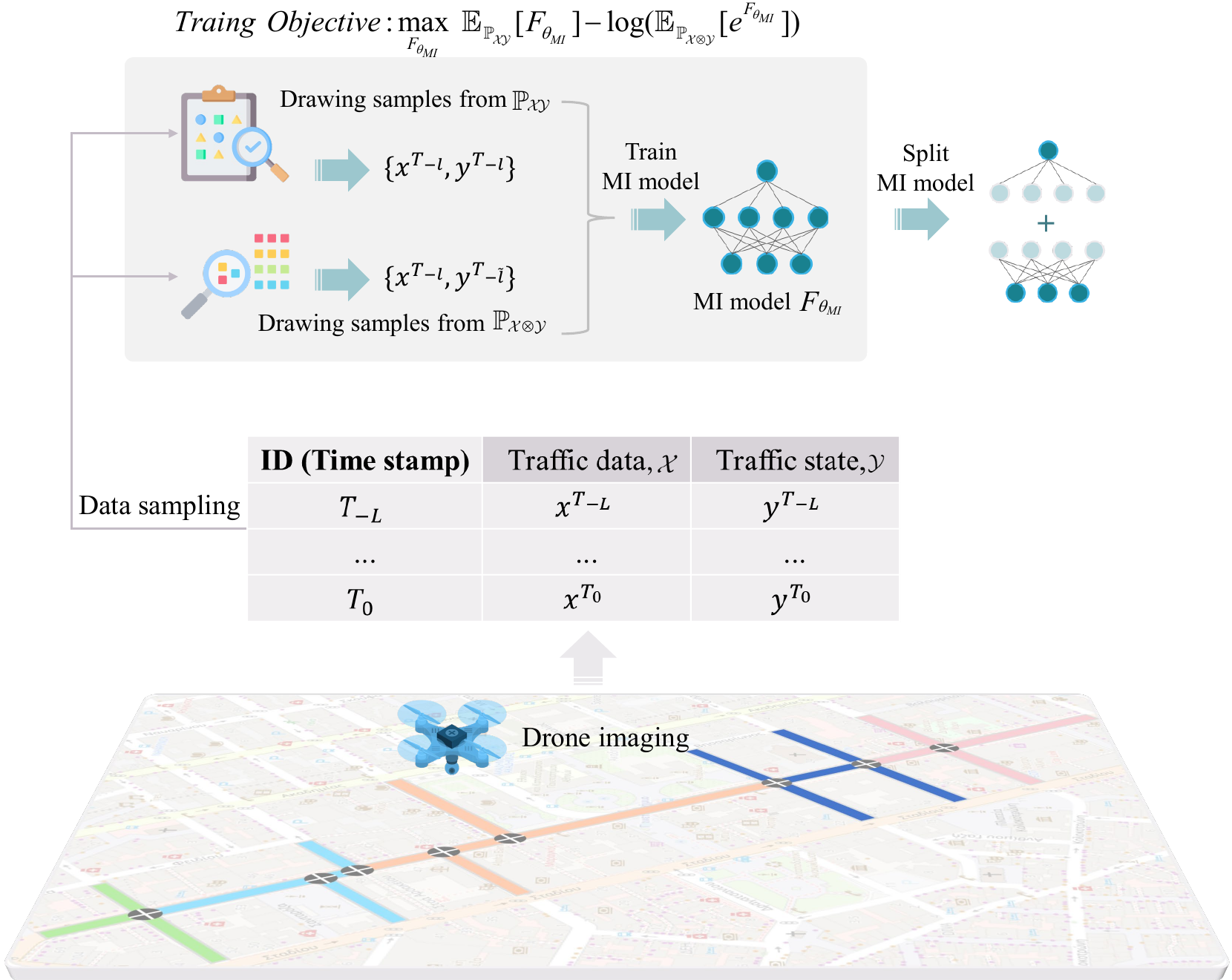}
        \caption{Workflow diagram of the mutual information (MI) model training, where the diagram selects the road segment labeled with orange as an example for better illustration.}
        \label{fig2}
    \end{center}
\end{figure}

By analyzing $\mathcal{P}2$, we found that $\mathcal{P}2$ seemingly can be addressed by the MI methodology proposed in \cite{jiang2022vf}. Nonetheless, the MI methodology used in \cite{jiang2022vf} is a KNN-based MI variant, which is only capable of measuring the MI value between discrete data and is restricted by the dimension of the data. To this end, we propose to unleash the power of neural networks to approximate the MI value between MP's data and the label possessed by MA.

Concretely, the data quality measurement function $Q(\cdot)$ will be substituted by the MI function $I(\cdot)$, i.e.,
\begin{equation} \label{eq:14}
    I(\mathcal{X}; \mathcal{Y}) = H(\mathcal{X}) - H(\mathcal{X}|\mathcal{Y}),
\end{equation}
where $H(\cdot)$ is the Shannon entropy. Here, the MI function can also be written in the form of divergence
\begin{equation} \label{eq:15}
     I(\mathcal{X}; \mathcal{Y}) = D_{KL}(\mathbb{P}_{\mathcal{X}\mathcal{Y}} \| \mathbb{P}_{\mathcal{X}} \otimes \mathbb{P}_{\mathcal{Y}}) = \mathbb{E}_{\mathbb{P}_{\mathcal{X}\mathcal{Y}}} \left[ \log \frac{d\mathbb{P}_{\mathcal{X}\mathcal{Y}}}{d\mathbb{P}_{\mathcal{X} \otimes \mathcal{Y}}} \right],
\end{equation}
where $\mathbb{P}_{\mathcal{X}\mathcal{Y}}$ and $\mathbb{P}_{\mathcal{X} \otimes \mathcal{Y}}$ are the joint and product of the marginal distribution, respectively.

Subsequently, as per the dual representation of the KL-divergence \cite{belghazi2018mutual}, the Donsker-Varadhan representation, Eq. (\ref{eq:15}) can be derived as
\begin{equation} \label{eq:16}
    I(\mathcal{X}; \mathcal{Y}) = \sup_{F : \Omega \rightarrow \mathbb{R}} \left( \mathbb{E}_{\mathbb{P}_{\mathcal{X}\mathcal{Y}}}[F] - \log(\mathbb{E}_{\mathbb{P}_{\mathcal{X} \otimes \mathcal{Y}}}[e^F]) \right),
\end{equation}
where $F$ is a measurable function mapping the data from sampling space $\Omega$ to the real number $\mathbb{R}$.

In consideration of that $\mathcal{F}$ is any class of functions $F: \Omega \rightarrow \mathbb{R}$ satisfying the constraints of the Donsker-Varadhan representation, Eq. (\ref{eq:16}) can be re-written as
\begin{equation} \label{eq:17}
    I(\mathcal{X}; \mathcal{Y}) \geq \sup_{F \in \mathcal{F}} \left( \mathbb{E}_{\mathbb{P}_{\mathcal{X}\mathcal{Y}}}[F] - \log(\mathbb{E}_{\mathbb{P}_{\mathcal{X} \otimes \mathcal{Y}}}[e^F]) \right).
\end{equation}

\begin{figure}[!t]
    \begin{center}
        \includegraphics[width=0.9\linewidth]{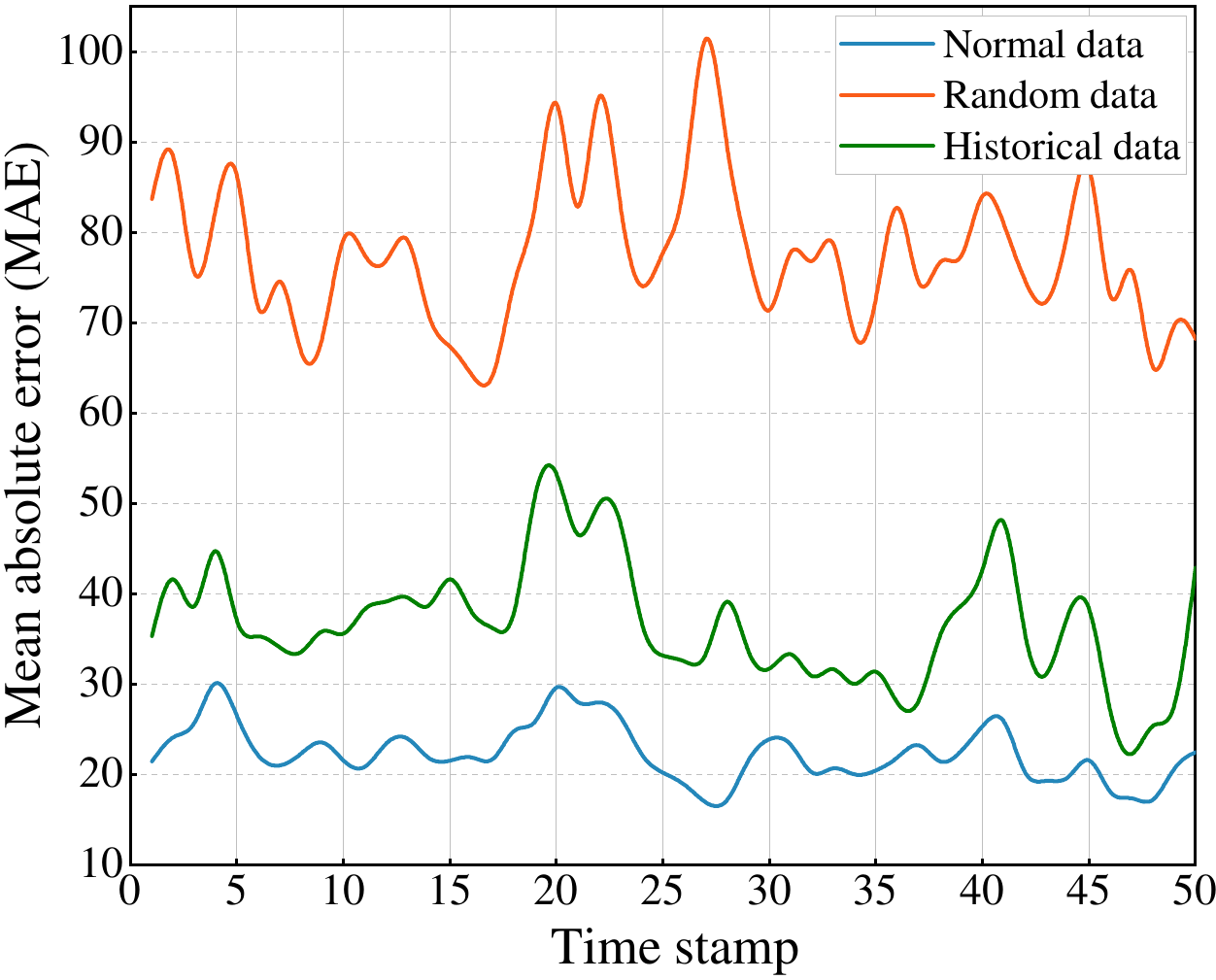}
        \caption{MAE comparison results versus varying the data quality supplied by MPs, where 2 out of 5 MPs are assumed to be lazy in traffic data provision for ``Historical data" and ``Random data".}
        \label{fig3}
    \end{center}
\end{figure}

Upon (\ref{eq:17}), we can utilize a deep neural network with parameters $\theta_{MI} \in \Theta_{MI}$ to represent $\mathcal{F}$, which indicates
\begin{equation} \label{eq:18}
    I(\mathcal{X}; \mathcal{Y}) \geq I_{\Theta_{MI}}(\mathcal{X}; \mathcal{Y}),
\end{equation}
where $\theta_{MI}$ is the parameters of the current neural network and $\Theta_{MI}$ is the collection of parameters of historical neural networks during the model training process.

Here, $I_{\Theta_{MI}}(\mathcal{X}; \mathcal{Y})$ is defined as
\begin{equation} \label{eq:19}
    I_{\Theta_{MI}}(\mathcal{X}; \mathcal{Y}) = \sup_{\theta_{MI} \in \Theta_{MI}} \left( \mathbb{E}_{\mathbb{P}_{\mathcal{X}\mathcal{Y}}}[F_{\theta_{MI}}] - \log(\mathbb{E}_{\mathbb{P}_{\mathcal{X} \otimes \mathcal{Y}}}[e^{F_{\theta_{MI}}}]) \right).
\end{equation}

To enhance the clarity, the workflow of the MI model training is presented in Fig. \ref{fig2}. Firstly, the MA utilizes drone imaging to collect traffic data on partitioned road segments, which means the MA needs to collect the data on each road segment for MI model training. The road segment labeled with orange is selected as an example in Fig. \ref{fig2}, where $L+1$ data is collected. Notably, $L$ does not need to be large, since the essence of the MI model is to capture the correlation factor between traffic data, $\mathcal{X}$, and the traffic label, $\mathcal{Y}$, which is analogous to a linear regression problem. Secondly, the MA randomly draws samples from the collected dataset, in which $x^{T_{-l}}$ and $y^{T_{-l}}$ in $\{x^{T_{-l}}, y^{T_{-l}}\}$ are sampled simultaneously, instead, $x^{T_{-l}}$ and $y^{T_{-\tilde{l}}}$ in $\{x^{T_{-l}}, y^{T_{-\tilde{l}}}\}$ are sampled separately and then concatenate them. Thirdly, the MA feeds sampled data into the MI model, $F_{\theta_{MI}}$, as input and trains the MI model intending to resolve (\ref{eq:19}), thereby acquiring the optimal MI model $F_{\theta_{MI}^{*}}$ for MI approximation. Lastly, the MA splits the well-trained MI model to retain the property of VFL. It is worth noting that differential privacy \cite{zhaoL2024Verti} and holomorphic encryption \cite{jiang2022vf} can be integrated into the split MI model to assure MPs' sample data privacy.

Upon the above description, the approximate MI between $\mathcal{X}_{k,n}$ and $\mathcal{Y}$ is defined as
\begin{equation} \label{eq:20}
    I_{\mathbb{n}}(\mathcal{X}_{k,n}; \mathcal{Y}) = \mathbb{E}_{\mathbb{P}^{\mathbb{n}}_{\mathcal{X}\mathcal{Y}}}[F^{*}_{\theta_{MI}}] - \log(\mathbb{E}_{\mathbb{P}^{\mathbb{n}}_{\mathcal{X} \otimes \mathcal{Y}}}[e^{F^{*}_{\theta_{MI}}}]),
\end{equation}
where the superscript $\mathbb{n}$ of ${\mathbb{P}^{\mathbb{n}}_{\mathcal{X}\mathcal{Y}}}$ and ${\mathbb{P}^{\mathbb{n}}_{\mathcal{X} \otimes \mathcal{Y}}}$ indicates the sampling count required to approximate the MI value.

\subsection{Problem Resolution} \label{sec:4.3}
As per the proposed MI model in Section \ref{sec:4.2}, we can approximate the data quality function $Q(\mathcal{X}_{k,n})$ with the utilization of (\ref{eq:20}), i.e.,
\begin{equation} \label{eq:21}
    Q(\mathcal{X}_{k,n}) \approx  I_{\mathbb{n}}(\mathcal{X}_{k,n}; \mathcal{Y}).
\end{equation}

To this end, the $\mathcal{P}2$ can be resolved with ease, and thus the most competent MP regarding each road segment can be selected. Subsequently, $K$ most competent MPs and MA collectively train the VFL-based TSE model, i.e., address the $\mathcal{P}3$. Here, we will utilize the method proposed in \cite{yu2017spatio} to train the VFL model, thereby addressing $\mathcal{P}3$.

\section{Penalty-Based Incentive Mechanism Design} \label{sec:5}
We first analyze the lazy behavior of MPs in data provision during real-time TSE and explore the countermeasures that the MA can implement in Section \ref{sec:5.1}. Subsequently, Section \ref{sec:5.2} presents the design of a penalty-based incentive mechanism along with the corresponding theoretical analysis.
\subsection{Problem Analysis} \label{sec:5.1}
By analyzing the utility model of MP in Section \ref{sec:3.2}, MP might be lazy in traffic data provision to augment the utility. Concretely, when MPs provide random generated or historical traffic data to MA for traffic state prediction, they can save the cost of data collection and processing. By doing so, however, the utility of the MA will be undermined, as the MAE comparison results depicted in Fig. \ref{fig3}. Therefore, certain countermeasures should be implemented to safeguard the utility of the MA.

\begin{table}[!t]
\caption{The Payoff Matrix of MA and MP}
\label{tab:1}
\begin{center}
    \renewcommand{\arraystretch}{1.5}
    \begin{tabular}{cccc}
    \hline
    \multicolumn{2}{c}{\multirow{2}{*}{}}                           & \multicolumn{2}{c}{MP}                                             \\ \cline{3-4} 
    \multicolumn{2}{c}{}                                            & Sloth ($\gamma$) & No Sloth ($1-\gamma$) \\ \hline
    \multirow{2}{*}{MA} & Inspection ($\epsilon$)      & (${\tilde{U}_{MA}},~\tilde{U}_{MP}$)                             & (${\bar{U}_{MA}},~ U_{MP}$)                                  \\ \cline{2-4} 
                        & No Inspection ($1-\epsilon$) & (${\hat{U}_{MA}},~\hat{U}_{MP}$)                             & (${U_{MA}},~ U_{MP}$)                                  \\ \hline
    \end{tabular}
\end{center}
\end{table}

Given the MI model proposed in Section \ref{sec:4} and the supervision mechanism suggested in \cite{fu2023incentive}, the MA can suppress the lazy probability of MPs by leveraging the well-trained MI model and penalty mechanism. Specifically, i) MA can inspect the traffic data provided by MP via well-trained MI models, i.e., (\ref{eq:20}), and (ii) a penalty will imposed on MP if MP is caught providing traffic data of low quality. Consequently, if a penalty is imposed on MP, i.e., MA finds that MP is lazy while data provision, the utility model of the MP will modified from (\ref{eq:8}) into
\begin{equation} \label{eq:22}
    \tilde{U}_{MP} = W - \tau_e \mathbb{T} E_r - \rho,
\end{equation}
where $\rho$ is the monetary of the penalty.

Conversely, if no penalty is imposed on MP while MP is lazy in data provision, i.e., MA does not supervise MP, thereby leading (\ref{eq:8}) transformed into
\begin{equation} \label{eq:23}
    \hat{U}_{MP} = W - \tau_e \mathbb{T} E_r.
\end{equation}

Analogously, if MP checks that MP is lazy while data provision, the utility model of the MA will shift from (\ref{eq:9}) to
\begin{equation} \label{eq:24}
    {\tilde{U}_{MA}} = \frac{1}{K}{\tau _a} \mathbb{I} (\tilde x - x^{\prime}) + \rho - S,
\end{equation}
where $S$ indicates the monetary cost of utilizing the MI model to supervise the behavior of MP in data provision and $x^{\prime}$ means degraded prediction accuracy of traffic states due to the lazy behavior of MPs in data provision.

If MP is lazy in data provision while MA does not inspect the MP, the utility model of the MA will change from (\ref{eq:9}) to
\begin{equation} \label{eq:25}
    {\hat{U}_{MA}} = \frac{1}{K}{\tau _a} \mathbb{I} (\tilde x - x^{\prime}) - W.
\end{equation}

Finally, if MP supplies data in quality while MA inspects the MP, the utility model of the MA will shift from (\ref{eq:9}) to
\begin{equation} \label{eq:26}
    {\bar{U}_{MA}} = \frac{1}{K}{\tau _a} \mathbb{I} (\tilde x - x) - W - S.
\end{equation}

\begin{table}[!t]
\caption{The Payoff Matrix of MA and MP with An Additional Penalty Imposed on MP}
\label{tab:2}
\begin{center}
    \renewcommand{\arraystretch}{1.5}
    \begin{tabular}{cccc}
    \hline
    \multicolumn{2}{c}{\multirow{2}{*}{}}                           & \multicolumn{2}{c}{MP}                                             \\ \cline{3-4} 
    \multicolumn{2}{c}{}                                            & Sloth ($\gamma$) & No Sloth ($1-\gamma$) \\ \hline
    \multirow{2}{*}{MA} & Inspection ($\epsilon$)      & (${\tilde{U}^{p}_{MA}},~\tilde{U}^{p}_{MP}$)                             & (${\bar{U}_{MA}},~ U_{MP}$)                                  \\ \cline{2-4} 
                        & No Inspection ($1-\epsilon$) & (${\hat{U}_{MA}},~\hat{U}_{MP}$)                             & (${U_{MA}},~ U_{MP}$)                                  \\ \hline
    \end{tabular}
\end{center}
\end{table}

To enhance clarity, we substitute $\frac{1}{K}{\tau _a} \mathbb{I} (\tilde x - x^{\prime})$ with $\pi^{\prime}$, $\frac{1}{K}{\tau _a} \mathbb{I} (\tilde x - x)$ with $\pi$, $\tau_e \mathbb{T} E_r$ with $H^{\prime}$, and $\tau_e \mathbb{T} (E_t + E_p + E_r)$ with $H$, respectively, in the remainder of this paper. Moreover, as per the property of (\ref{eq:10}) and the results presented in Fig. \ref{fig3}, we assume $\pi^{\prime} = 0$ with well-designed TSE performance threshold $\tilde{x}$. 

Upon the above analysis, we can model the utility model of MA and MP in a game form, i.e., Table \ref{tab:1}. Analyzing Table \ref{tab:1}, the expected utility of MA with strategies of inspection and no inspection are defined as
 
\begin{equation} \label{eq:27}
    \mathbb{E}_{i}({U_{MA}}) = \gamma \tilde{U}_{MA} + (1-\gamma) \bar{U}_{MA}
\end{equation}
and
\begin{equation} \label{eq:28}
    \mathbb{E}_{ni}({U_{MA}}) = \gamma \hat{U}_{MA} + (1-\gamma) U_{MA}.
\end{equation}

According to the definition of equilibrium, when MP calibrates the sloth probability $\gamma$ to enable $\mathbb{E}_{i}({U_{MA}}) = \mathbb{E}_{ni}({U_{MA}})$, the equilibrium result of MP can be acquired, since MA cannot improve its utility by unilaterally changing the strategy. Therefore, the optimal sloth probability of MP is defined as
\begin{equation} \label{eq:29}
    \gamma^* = \frac{S}{\rho+W}.
\end{equation}

Similarly, the optimal inspection probability of MA is derived as
\begin{equation} \label{eq:30}
    \epsilon^* = \frac{H - H^{\prime}}{\rho+W},
\end{equation}
where
\begin{equation} \label{eq:31}
    H - H^{\prime} = \tau_e \mathbb{T} (E_t + E_p).
\end{equation}

Observing (\ref{eq:29}) and (\ref{eq:30}), we can find that the inspection and sloth probability of MA and MP are inversely proportional to the scale of the penalty $\rho$ and reward $W$. In other words, if $\rho$ or $W$ is set as infinity, the inspection and sloth probability of MA and MP will approach zero, thereby resolving Q2. Nonetheless, two consequences might occur, i) if $\rho$ is too large, MP might be deterred to participate in traffic data provision; ii) $W$ might exceed the budget of MA. To this end, in the next section, we will provide a penalty-based incentive mechanism to circumvent the aforementioned two consequences while tackling Q2.

\subsection{Incentive Design and Theoretical Analysis} \label{sec:5.2}

Building on the analysis in Section \ref{sec:5.1} and drawing inspiration from \cite{fu2023incentive}, we propose an MI model empowered and penalty-based incentive mechanism. Specifically, if an MP under the supervision of the MA demonstrates sloth in providing traffic data for two consecutive times, an additional penalty of $(\beta -1) \rho$ will be imposed on the MP. Notably, the supervision of MP's behavior in traffic data provision is conducted via the MI model proposed in Section \ref{sec:4.2}. Consequently, the utility function of MP and MA will change from (\ref{eq:22}) and (\ref{eq:24}) to
\begin{equation} \label{eq:32}
    \tilde{U}^{p}_{MP} = W - H^{\prime} - \beta \rho
\end{equation}
and
\begin{equation} \label{eq:33}
    {\tilde{U}^{p}_{MA}} = \beta \rho - S,
\end{equation}
respectively. Accordingly, the payoff matrix of MA and MP will also shift from Table \ref{tab:1} to \ref{tab:2}.

Analyzing our proposed incentive mechanism, two cases exist for each game interaction between MA and MP, except the first game interaction, which are

(i) If MA detects MP that it is lazy while traffic data provision at the $i$-th settlement cycle, the probability is $\gamma_i \epsilon_i$, thereby causing the payoff matrix of MA and MP at the $i+1$-th settlement cycle is the in form of Table \ref{tab:2}.

(ii) Otherwise, for other cases, i.e., MA does not detect that MP is lazy while traffic data provision at the $i$-th settlement cycle, the probability can be expressed as $1 - \gamma_i \epsilon_i$, leading to the payoff matrix of MA and MP at the $i+1$-th settlement cycle is in the form of Table \ref{tab:1}.

Upon the above analysis, the expected utility function of MA at the $i+1$-th settlement cycle is defined as
\begin{equation} \label{eq:34}
    \begin{aligned}
        \mathbb{E}_{i+1}(U_{MA}) = &(1 - \gamma_i \epsilon_i) (\gamma_{i+1} \epsilon_{i+1} \tilde{U}_{MA}) + \gamma_i \epsilon_i (\gamma_{i+1} \epsilon_{i+1} \tilde{U}^p_{MA}) \\ &+ (1-\gamma_{i+1})\epsilon_{i+1} \bar{U}_{MA} + \gamma_{i+1}(1-\epsilon_{i+1})\hat{U}_{MA} \\& + (1-\gamma_{i+1})(1-\epsilon_{i+1})U_{MA}
    \end{aligned}
\end{equation}

Analogous to the analysis for deriving (\ref{eq:29}), the optimal sloth probability of MP at the $i+1$-th settlement cycle can be derived by tackling
\begin{equation} \label{eq:35}
    \frac{\partial \mathbb{E}_{i+1}(U_{MA})}{\partial \epsilon_{i+1}} = S - \gamma_{i+1}(\gamma_i \epsilon_i \rho (\beta - 1) + \rho) = 0.
\end{equation}

As per (\ref{eq:35}), we can obtain
\begin{equation} \label{eq:36}
    \gamma_{i+1} = \frac{S}{\gamma_i \epsilon_i \rho (\beta - 1) + \rho}.
\end{equation}

Notably, from MP and MA's perspective, at the $i+1$-th settlement cycle, due to the Markov assumption, we assume that settlement cycles before the $i-1$-th one will not affect the strategy formulation of MP and MA, i.e., the $i$-th settlement cycle is the first game interaction between MA and MP. To this end, the optimal strategy of MP at the $i$-th settlement cycle is aligned with (\ref{eq:29}), which means
\begin{equation} \label{eq:37}
    \gamma_i = \gamma^* =  \frac{S}{\rho+W}.
\end{equation}

Observing (\ref{eq:36}) and (\ref{eq:37}), we can find that
\begin{equation} \label{eq:38}
    \begin{aligned}
        \frac{1}{\gamma_{i+1}} &= \frac{\gamma_i \epsilon_i \rho (\beta - 1) + \rho}{S} \\
        &= \frac{\gamma_i \epsilon_i \rho (\beta - 1) + \rho + W - W}{S} \\
        &= \frac{1}{\gamma_{i}} + \frac{\gamma_i \epsilon_i \rho (\beta - 1)}{S} - \frac{W}{S}.
    \end{aligned}
\end{equation}

Analogously, the expected utility function of MP at the $i+1$-th settlement cycle is defined as
\begin{equation} \label{eq:39}
    \begin{aligned}
        \mathbb{E}_{i+1}(U_{MP}) = &(1 - \gamma_i \epsilon_i) (\gamma_{i+1} \epsilon_{i+1} \tilde{U}_{MP}) + \gamma_i \epsilon_i (\gamma_{i+1} \epsilon_{i+1} \tilde{U}^p_{MP}) \\ &+ (1-\gamma_{i+1}) U_{MP} + \gamma_{i+1}(1-\epsilon_{i+1})\hat{U}_{MP}
    \end{aligned}
\end{equation}

Based on (\ref{eq:39}), the optimal inspection probability of MA at the $i+1$-th settlement cycle can also be derived by tackling
\begin{equation} \label{eq:40}
    \frac{\partial \mathbb{E}_{i+1}(U_{MP})}{\partial \gamma_{i+1}} = H - H^{\prime} - \epsilon_{i+1}(\gamma_i \epsilon_i \rho (\beta - 1) + \rho) = 0,
\end{equation}
which leads to the solution
\begin{equation} \label{eq:41}
    \epsilon_{i+1} = \frac{H - H^{\prime}}{\gamma_i \epsilon_i \rho (\beta - 1) + \rho}.
\end{equation}

Then, the deduction of the relationship between $\epsilon_{i+1}$ and $\epsilon_i$ is akin to the analysis for deriving (\ref{eq:37}) and (\ref{eq:38}). We can derive that
\begin{equation} \label{eq:42}
    \epsilon_i = \epsilon^* =  \frac{H - H^{\prime}}{\rho+W},
\end{equation}
thereby yielding
\begin{equation} \label{eq:43}
    \begin{aligned}
        \frac{1}{\epsilon_{i+1}} &= \frac{\gamma_i \epsilon_i \rho (\beta - 1) + \rho}{H - H^{\prime}} \\
        &= \frac{\gamma_i \epsilon_i \rho (\beta - 1) + \rho + W - W}{H - H^{\prime}} \\
        &= \frac{1}{\epsilon_{i}} + \frac{\gamma_i \epsilon_i \rho (\beta - 1)}{H - H^{\prime}} - \frac{W}{H - H^{\prime}}.
    \end{aligned}
\end{equation}

Upon (\ref{eq:38}) and (\ref{eq:43}), we can deduce that the sloth probability of MP, $\gamma$, and the inspection probability of MA, $\epsilon$, in traffic data provision will keep diminishing as the settlement cycle progresses unless either of the following conditions reaches
\begin{equation} \label{eq:44}
\left\{ \begin{array}{l}
{\gamma_i \epsilon_i \rho (\beta - 1) = W,}\\
{\gamma_i} = 0,\\
{\epsilon_i} = 0.
\end{array} \right.
\end{equation}

With the assumption that at least one condition in (\ref{eq:44}) is reached in the $i^*$-th settlement cycle, we can derive Proposition \ref{prop:1} to inscribe the theoretical upper bound of $\gamma$.

\begin{proposition} \label{prop:1}
    With the game interaction between MA and MP progresses, the finalized optimal sloth probability of MP, $\gamma$, is bounded as
    \begin{equation} \label{eq:45}
        \gamma < \frac{1}{\frac{1}{\gamma_{1}} + \frac{(i^{*} - 1)}{S} \left[ \gamma_{i^* - 1} \epsilon_{i^* - 1} \rho (\beta - 1) - W \right]}
    \end{equation}
\end{proposition}

\begin{proof}
    As per (\ref{eq:38}), we can derive

    \begin{align}
        \frac{1}{\gamma_{i^*}} &= \frac{1}{\gamma_{i^* - 1}} + \frac{\gamma_{i^* - 1} \epsilon_{i^* - 1} \rho (\beta - 1)}{S} - \frac{W}{S}, \label{eq:46a} \\
        \frac{1}{\gamma_{i^* - 1}} &= \frac{1}{\gamma_{i^* - 2}} + \frac{\gamma_{i^* - 2} \epsilon_{i^* - 2} \rho (\beta - 1)}{S} - \frac{W}{S}, \label{eq:46b} \\
        &\cdots \notag \\
        \frac{1}{\gamma_{2}} &= \frac{1}{\gamma_{1}} + \frac{\gamma_{1} \epsilon_{1} \rho (\beta - 1)}{S} - \frac{W}{S}. \label{eq:46c}
    \end{align}

    By consolidating (\ref{eq:46a}), (\ref{eq:46b}), and (\ref{eq:46c}), we can express $\gamma_{i^*}$ as
    \begin{equation} \label{eq:49}
        \begin{aligned}
            \frac{1}{\gamma_{i^*}} = & \frac{\rho (\beta - 1)}{S} \left[\gamma_{i^* - 1} \epsilon_{i^* - 1} + \gamma_{i^* - 2} \epsilon_{i^* - 2} + \cdots + \gamma_{1} \epsilon_{1} \right] \\ &+ \frac{1}{\gamma_{1}} - (i^{*} - 1)\frac{W}{S}
        \end{aligned}
    \end{equation}

    Subsequently, upon the definition of $\gamma_{i^*}$ and (\ref{eq:44}), we can deduce
    \begin{equation} \label{eq:50}
        \left\{ \begin{array}{l}
        {\gamma_{i^* -1} \epsilon_{i^* -1} \rho (\beta - 1) > W,}\\
        {\gamma_{i^* -1}} > 0,\\
        {\epsilon_{i^* -1}} > 0.
        \end{array} \right.
    \end{equation}

    Substituting the conditions presented in (\ref{eq:50}) into (\ref{eq:46a}), we can derive the below inequality
    \begin{equation} \label{eq:51}
            \frac{1}{\gamma_{i^*}} > \frac{1}{\gamma_{i^* - 1}} \Leftrightarrow \gamma_{i^*} < \gamma_{i^* - 1}.
    \end{equation}

    In parallel, according to (\ref{eq:43}) and (\ref{eq:50}), we can deduce
    \begin{equation} \label{eq:52}
        \frac{1}{\epsilon_{i^*}} > \frac{1}{\epsilon_{i^* - 1}} \Leftrightarrow \epsilon_{i^*} < \epsilon_{i^* - 1}.
    \end{equation}

    With the utilization of (\ref{eq:50}), (\ref{eq:51}), and (\ref{eq:52}), we can derive
    \begin{equation} \label{eq:53}
        \gamma_{i^*} \epsilon_{i^*} < \gamma_{i^* -1} \epsilon_{i^* -1}.
    \end{equation}
    Moreover, by processing (\ref{eq:46b}) and (\ref{eq:46c}) in a resemble manner of (\ref{eq:50}), (\ref{eq:51}), and (\ref{eq:52}), we can deduce
    \begin{equation} \label{eq:54}
        \gamma_{i^*} \epsilon_{i^*} < \gamma_{i^* -1} \epsilon_{i^* -1} < \cdots < \gamma_{1} \epsilon_{1}.
    \end{equation}

    Consequently, substituting (\ref{eq:54}) into (\ref{eq:49}), we can derive the inequality
    \begin{equation} \label{eq:55}
        \begin{aligned}
            \frac{1}{\gamma_{i^*}} >& \frac{1}{\gamma_{1}} + \frac{\rho (\beta - 1)}{S} \left[ (i^{*} - 1) \gamma_{i^* - 1} \epsilon_{i^* - 1} \right] - (i^{*} - 1)\frac{W}{S} \\
            =& \frac{1}{\gamma_{1}} + \frac{(i^{*} - 1)}{S} \left[ \gamma_{i^* - 1} \epsilon_{i^* - 1} \rho (\beta - 1) - W \right].
        \end{aligned}
    \end{equation}

    Therefore,
    \begin{equation} \label{eq:56}
        \gamma_{i^*} < \frac{1}{\frac{1}{\gamma_{1}} + \frac{(i^{*} - 1)}{S} \left[ \gamma_{i^* - 1} \epsilon_{i^* - 1} \rho (\beta - 1) - W \right]}.
    \end{equation}
\end{proof}

Analogously, we propose Proposition \ref{prop:2} to inscribe the theoretical upper bound of $\epsilon$. 
\begin{proposition} \label{prop:2}
    With the game interaction between MA and MP progresses, the finalized optimal inspection probability of MA, $\epsilon$, is bounded as
    \begin{equation} \label{eq:57}
        \epsilon < \frac{1}{\frac{1}{\epsilon_{1}} + \frac{(i^{*} - 1)}{H - H^{\prime}} \left[ \gamma_{i^* - 1} \epsilon_{i^* - 1} \rho (\beta - 1) - W \right]}.
    \end{equation}
\end{proposition}

\begin{proof}
    The proof is in the same vein as Proposition \ref{prop:1} and is therefore omitted.
\end{proof}

\textbf{Remarks:} Several benefits of our proposed incentive mechanism exist:
\begin{itemize}
    \item[1.] The budget constraint of MA will not be breached.
    \item[2.] The probability that MP is deterred from traffic data provision is diminished.
    \item[3.] The interaction between MA and MP can be modeled as a Markov process, meaning that the outcome of the current round depends only on the result of the previous round.
    \item[4.] A consistent derivation between $(\gamma_i, \epsilon_i)$ and $(\gamma_{i+1}, \epsilon_{i+1})$ can be obtained, so as to facilitate acquiring a closed-formed incentive scheme.
\end{itemize}

\begin{table}[!t] 
\centering
\caption{Table of key parameters for incentive mechanism}
\label{tab:3}
\begin{tabular}{|ll|}
\hline
\multicolumn{1}{|c|}{Simulation Parameters} & \multicolumn{1}{c|}{Value} \\ \hline
\multicolumn{2}{|c|}{Data Collection Parameters}                         \\ \hline
\multicolumn{1}{|l|}{$B$}                   & $10$MHz                        \\ \hline
\multicolumn{1}{|l|}{$p$}                   & $0.1$ W                       \\ \hline
\multicolumn{1}{|l|}{$d(t)$}                & $[100 \sim 500]$ m  \\ \hline
\multicolumn{1}{|l|}{$N_0$}                 & $-174$ dBm/Hz                         \\ \hline
\multicolumn{1}{|l|}{$D$}                   & $80$ Bytes                    \\ \hline
\multicolumn{1}{|l|}{$\Delta t$}            & $0.4$ s                        \\ \hline
\multicolumn{2}{|c|}{Data Processing Parameters}                         \\ \hline
\multicolumn{1}{|l|}{$\eta$}                & $10^{-26}$                        \\ \hline
\multicolumn{1}{|l|}{$c$}                   & $1.5 \times 10^{4}$                       \\ \hline
\multicolumn{1}{|l|}{$f_c$}                 & $10^{9}$ Hz                  \\ \hline
\multicolumn{2}{|c|}{Sub-model Running Parameters}                        \\ \hline
\multicolumn{1}{|l|}{$E_r$}                 & $7.2 \times 10^{-4}$ J      \\ \hline
\multicolumn{1}{|l|}{$\tau_e$}              & $2.44 \times 10^{-4}$       \\ \hline
\multicolumn{1}{|l|}{$\mathbb{T}$}          & $8,640$                        \\ \hline
\multicolumn{2}{|c|}{MA's Profit Parameters}                              \\ \hline
\multicolumn{1}{|l|}{$\tilde{x}$}           & $25$                    \\ \hline
\multicolumn{1}{|l|}{$\tau_a$}              & $1,000$                 \\ \hline
\multicolumn{1}{|l|}{$W$}                   & $150$                       \\ \hline
\multicolumn{1}{|l|}{$S$}                   & $300$                       \\ \hline
\end{tabular}
\end{table}

\section{Experimental Evaluation} \label{sec:6}
We begin with introducing the experimental configurations in Section \ref{sec:6.1}. Then, we assess proposed MI models in data provider selection by varying hyperparameters and compare the performance of TSE with selected baselines in Section \ref{sec:6.2}. Lastly, we evaluate the performance of the MI model-empowered incentive mechanism in Section \ref{sec:6.3}.

\subsection{Experimental Configurations} \label{sec:6.1}
\subsubsection{\textbf{Datasets and settings of the VFL-based TSE model}} \label{sec:6.1.1}
Regarding the dataset for the VFL-based TSE model training and validation, a real-world dataset, pNEUMA \cite{barmpounakis2020new}, is utilized, which includes the trajectories of vehicles on the morning peak time of four weekdays in the central district of Athens, Greece. Notably, the pNEUMA is collected with a fleet of ten drones, similar to \cite{wang2024privacy}, we utilize the second, third, and fifth drone-collected trajectories data to assess the performance of our proposed framework. By doing so, the urban transportation network $\mathcal{G}$ includes 8 intersections and 21 road links, the specific geometric layout is presented at the bottom of Fig. \ref{fig2}. To align with the illustration of Fig. \ref{fig2}, we consider that there are $5$ road segments, each road segment is labeled with a distinct color, and $6$ MPs supply the traffic data for each road segment, i.e., $K=5$ and $N=6$. In addition, the traffic states that we are concerned about in this paper are traffic flow and traffic density of each road. We utilize the open-source map-matching software LeuvenMapMatching \cite{meert2018hmm} to transform the vehicle trajectories into traffic flow and traffic density.

For the VFL model, we assume the MA constructs a three-layer Multilayer Perceptron (MLP), i.e., the top sub-model of the VFL model, to fuse the intermediate results uploaded by MPs and each MP utilizes a Spatio-Temporal Convolutional Network (STGCN) proposed in \cite{yu2017spatio}, the bottom sub-model of the VFL model, to process traffic states collected on its corresponding road segment. In addition, we assume the data sampling interval of the VFL model is set as $\Delta T = 10$ seconds, as well as the MA and MPs will leverage the past 9 time steps to predict the traffic states at the current time step, which indicates $\tau_i = 9$ and $\tau_{o} = 1$. The learning rate for the bottom sub-model STGCN and the top sub-model MLP are all set to 3e-4. Eventually, for VFL-based TSE model training and validation, we split the data as per the Pareto principle, i.e., 80\% data for training and 20\% used for validation, and root mean square error (RMSE) and mean absolute error (MAE) are opted as evaluation criteria.

\begin{table}[!t]
\caption{Averaged MI values of 5 road segments with varying the mean and std of the injected Gaussian Noise}
\label{tab:4}
\begin{center}
    \renewcommand{\arraystretch}{1.35} 
    \begin{tabular}{|c|cccccc|}
    \hline
     \diagbox{$\mu$}{$\sigma$}       & $0$ & $0.01$  & $0.05$           & $0.1$ & $0.2$    & $0.3$                  \\ \hline
    \multicolumn{1}{|c|}{$0$}    & 1.112         & 1.11             & 1.058            & 0.902           & 0.412           & -0.184          \\ \hline
    \multicolumn{1}{|c|}{$0.01$} & 1.062         & 1.062            & 1.028            & 0.896           & 0.414           & -0.196          \\ \hline
    \multicolumn{1}{|c|}{$0.05$} & 0.434         & 0.412            & 0.312            & 0.16            & -0.208          & -0.648          \\ \hline
    \multicolumn{1}{|c|}{$0.1$}  & -1.404        & -1.424           & -1.492           & -1.568          & -1.732          & -1.946          \\ \hline
    \multicolumn{1}{|c|}{$0.2$}  & -6.512        & -6.518           & -6.546           & -6.556          & -6.504          & -6.372          \\ \hline
    \multicolumn{1}{|c|}{$0.3$}  & -11.174       & -11.152          & -11.06           & -10.962         & -10.854         & -10.854         \\ \hline
    \end{tabular}
\end{center}
\end{table}

\subsubsection{\textbf{Settings for data provider selection}} \label{sec:6.1.2}
Regarding each road segment, we consider the traffic data composed of 300 continuous time steps was used for MI model training, the network architecture of each MI model is a three-layer MLP and the learning rate is set as 1e-2. To emulate the varying data quality across multiple MPs, we inject the Gaussian Noise with distinct mean and std values into the original data. Specifically, the mean and std values, denoted as $\mu$ and $\sigma$, of the Gaussian Noise are in the set of $\{0, 0.01, 0.05, 0.1, 0.2, 0.3\}$. In addition, the number of samples drawn for MI estimation is set as $50$.

To demonstrate the effectiveness of our proposed MI models-driven data provider selection scheme, several representative benchmarks are selected:
\begin{enumerate}
    \item[1.] Central: the TSE model is trained in a centralized mode, which can be deemed as the upper bound of the model trained in the setting of VFL.
    \item[2.] Random: The data provider participating in the VFL model training is opted randomly, which can be deemed as the lower bound of the model trained in the setting of VFL.
    \item[3.] Oracle \cite{wang2024privacy}: The optimal MPs are selected to train the VFL model, i.e., selecting the MP with traffic data that is injected with minimum noise for each road segment.
    \item[4.] VFLFS \cite{feng2022vertical}: The data provider participating in the VFL model training is selected during the model training process.
\end{enumerate}

\subsubsection{\textbf{Settings for incentive mechanism}} \label{sec:6.1.3}
The parameters used in the incentive mechanism design primarily stem from the utility models of MP and MA. Specifically, the MP utility model encompasses three components: data collection, data processing, and the sub-model running, while the MA utility model includes only the profit component. Accordingly, by referring to \cite{fu2023incentive, kazmi2021novel, ning2020intelligent, kim2017dual, lim2020hierarchical} and the data collected from the VFL-based TSE model training, we define all parameters in Table \ref{tab:3}, organized by these four components.

The benchmarks opted for validating our proposed incentive mechanism is listed as:
\begin{enumerate}
    \item[1.] BCL \cite{fu2020vehicular}: The collective learning without incentive mechanism.
    \item[2.] SGF: A supervision game-based incentive mechanism proposed in \cite{fu2023incentive} for federated learning.
\end{enumerate}

\begin{figure}[!t]
    \begin{center}
        \includegraphics[width=1\linewidth]{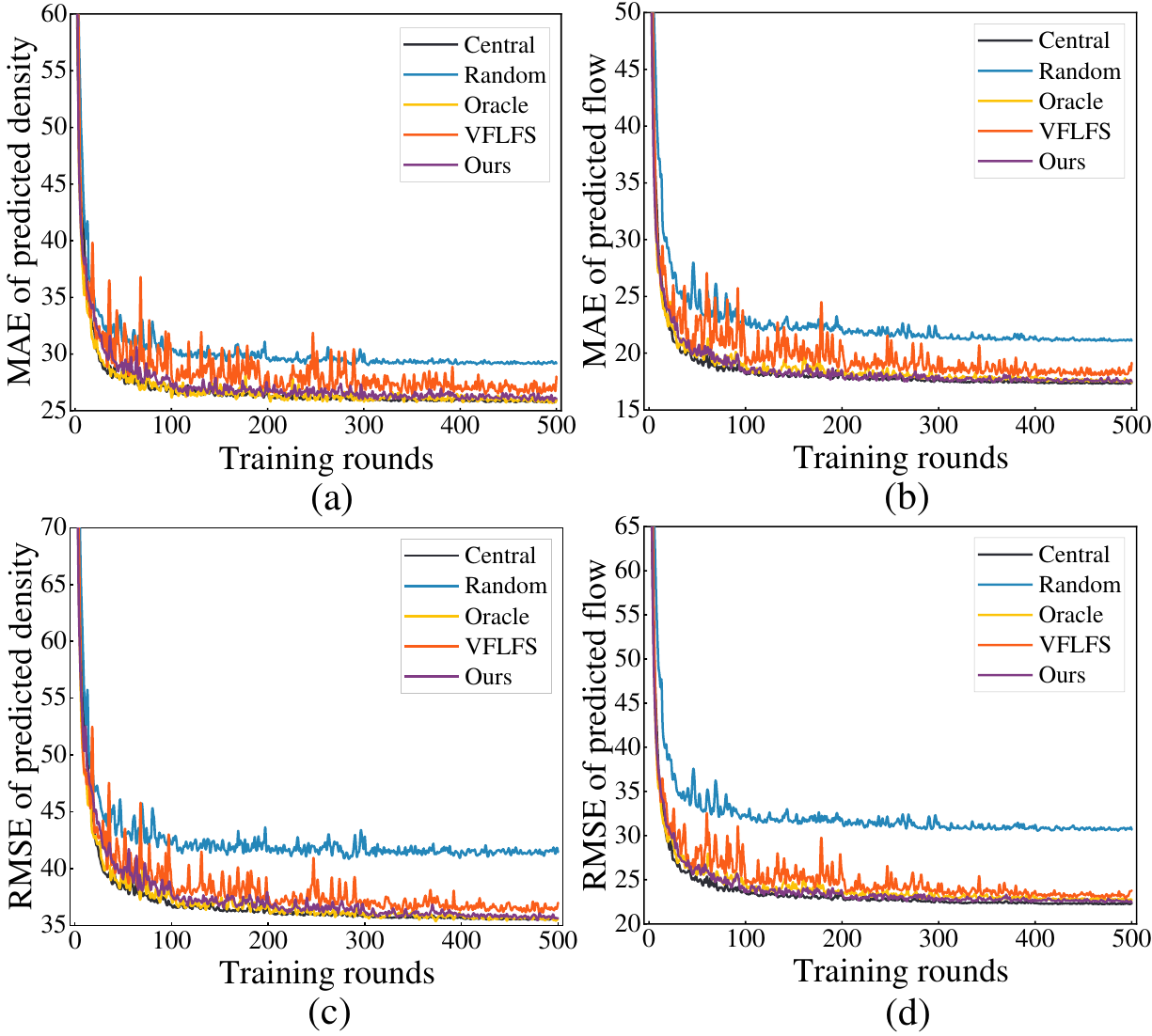}
        \caption{Training curve of VFL-based TSE model among our proposed method and benchmarks, where the loss in (a), (b), (c), and (d) are MAE of traffic density, MAE of traffic flow, RMSE of traffic density, and RMSE of traffic flow, respectively.}
        \label{fig4}
    \end{center}
\end{figure}

\subsection{Data Provider Selection Scheme Evaluation - Q1} \label{sec:6.2}
In this section, we will evaluate the effectiveness of our proposed MI models in MI reckon for the data provider, MP, selection. Subsequently, we will compare our proposed method with benchmarks in terms of VFL model training and traffic state prediction.

\subsubsection{\textbf{Effectiveness of MI models}} Upon our proposed MI value reckon method in Section \ref{sec:4.2}, we present the calculated MI values among the data injected with distinct Gaussian noise in Table \ref{tab:4}. Observing Table \ref{tab:4}, the predicted MI value is decreasing along with the increasing of $\mu$ and $\sigma$ of injected noise, i.e., the traffic data collected by MP with inferior quality will render a low MI score, thereby effectively facilitating the MP selection on each road segment.

\subsubsection{\textbf{Comparison with benchmarks}} As shown in Fig. \ref{fig4}, in terms of 4 selected loss metrics, our proposed method can achieve near-identical performance with the benchmark of Oracle and Central, which demonstrates the capability of our proposed method in data provider selection for efficient VFL model training. Notably, while VFLFS can also converge to the analogous performance of Oracle and Central, the fluctuation of the training curve is more dramatic than that of our proposed method.

\begin{table}[!t]
\caption{The comparison results between our proposed data selection scheme with benchmarks in traffic state prediction}
\label{tab:5}
\begin{center}
    \renewcommand{\arraystretch}{1.35}
    \begin{tabular}{c|ccc}
    \hline
    Methods                  & Evaluation Metrics & Traffic density & Traffic flow \\ \hline
    \multirow{2}{*}{Central} &  MAE     &{\textbf{25.735}}          & \textbf{17.357}       \\ \cline{2-4} 
                             &  RMSE    & {\textbf{35.358}}          & \textbf{22.308}       \\ \hline
    \multirow{2}{*}{Random}  &  MAE     & {29.465}          & 21.100       \\ \cline{2-4} 
                             &  RMSE    & {42.015}          & 30.766       \\ \hline
    \multirow{2}{*}{Oracle}   & MAE     & {\underline{26.065}}          & 17.621       \\ \cline{2-4} 
                              & RMSE    & {35.747}          & \underline{22.434}       \\ \hline
    \multirow{2}{*}{VFLFS}   &  MAE     & {28.019}          & 19.159       \\ \cline{2-4} 
                             &  RMSE    & {36.992}          & 23.776       \\ \hline
    \multirow{2}{*}{Our}     &  MAE     & {26.146}          & \underline{17.498}       \\ \cline{2-4} 
                             &  RMSE    & {\underline{35.637}}          & 23.006       \\ \hline
    \end{tabular}
\end{center}
\end{table}

Observing Table \ref{tab:5}, in terms of traffic density and traffic flow estimation, the Central method consistently achieves the best performance (labeled with \textbf{bold font}). It is worth noting that our proposed solution can attain analogous performance to the Oracle method, in which each method achieves 2 out of 4 second-best performance (labeled with the \underline{underline}). In addition, our proposed method can augment the prediction performance in terms of traffic density and traffic flow by 11.23\% and 21.15\% in comparison with the Random approach, respectively.

\subsection{Incentive Mechanism Evaluation - Q2} \label{sec:6.3}
\begin{table}[!t]
\caption{The average MI values across five road segments, obtained when the MP employs varying lazy behaviors and lazy percentages, are based on a sample count of 50 for the MI models.}
\label{tab:6}
\begin{center}
    \renewcommand{\arraystretch}{1.35}
    \begin{tabular}{c|ccccc}
    \hline
               & 20\%  & 40\%  & 60\%  & 80\%  & 100\% \\ \hline
    Random     & -1.53 & -1.97 & -6.34 & -9.1  & -9.24 \\ \hline
    Historical & -0.26 & -3.04 & -5.4  & -6.02 & -9.62 \\ \hline
    \end{tabular}
\end{center}
\end{table}

In this section, we first assess the effectiveness of our proposed method in identifying MPs' lazy behaviors while traffic data provision, which is crucial for our proposed penalty-based incentive mechanism. Next, we explore the strategies of MA and MP under different hyperparameter settings, in which the hyperparameter contains the MP's initial lazy probability $\gamma_0$, additional penalty coefficient $\beta$, and the base penalty $\rho$. Subsequently, we analyze the gap between our simulated numerical results and the theoretical bound. Lastly, we present the comparison results of our proposed incentive mechanism and benchmarks.

\subsubsection{\textbf{MP's lazy behaviors identify}} \label{sec:6.3.1}
Given that MPs may not exhibit consistent laziness during traffic data provision and may employ various lazy strategies, we assess the effectiveness of our proposed MI models across diverse conditions, with results presented in Table \ref{tab:6}. In this experiment, we assume that MPs may adopt two lazy strategies: using either randomly generated data or historical data for traffic state prediction. Additionally, we account for varying degrees of laziness among MPs, meaning different percentages of MP-held data are substituted with either randomly generated or historical data. As demonstrated in Table \ref{tab:6}, regardless of the lazy strategy or the degree of lazy behavior employed by MPs, the calculated MI values are consistently negative, demonstrating the efficacy of the MI model-empowered inspection scheme.

\subsubsection{\textbf{Impact of hyperparameters}} \label{sec:6.3.2}

We consider that both MA and MP will update their strategy based on their historical game results since both MA and MP are under information asymmetry in their counterparts' utility models. Specifically, we assume that the initial sloth probability of MP $\gamma_0$ and the initial inspection probability of MA $\epsilon_0$ are predefined as a value. Then, $\gamma_i$ and $\epsilon_i$ will update as per the cumulative payoff of MA and MP as the game evolves. Moreover, under varying settings, we explore the impact of the additional penalty coefficient $\beta$ in Figs. \ref{fig5} and \ref{fig6} and the impact of the base penalty $\rho$ in Figs. \ref{fig7} and \ref{fig8}.

\begin{figure}[!t]
    \begin{center}
        \includegraphics[width=1\linewidth]{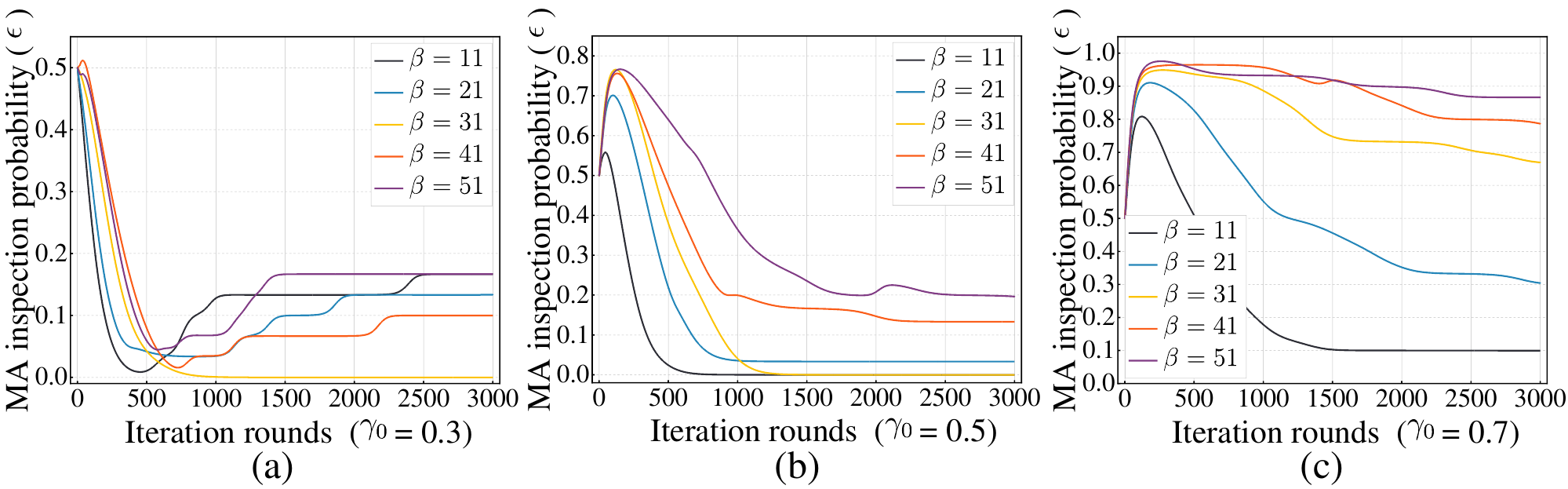}
        \caption{The inspection probability of MA $\epsilon$ versus the additional penalty coefficient $\beta$, where the initial inspection probability of MA of (a), (b), and (c) are all set as 0.5 while the initial sloth probability of MP is set as 0.3, 0.5, and 0.7.}
        \label{fig5}
    \end{center}
\end{figure}

\begin{figure}[!t]
    \begin{center}
        \includegraphics[width=1\linewidth]{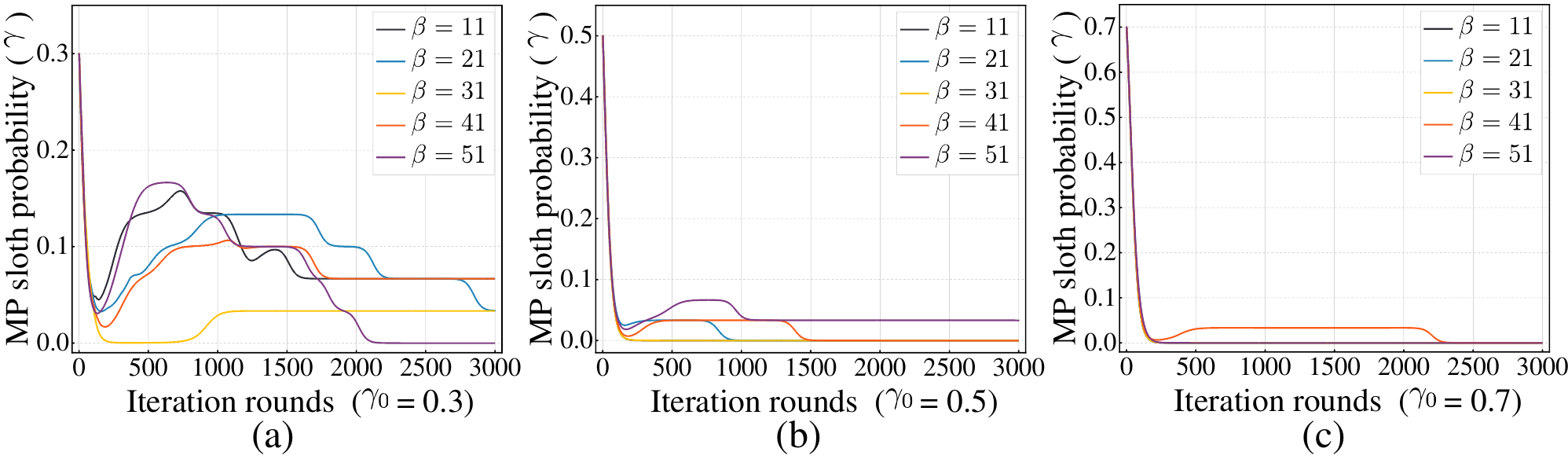}
        \caption{The sloth probability of MP $\gamma$ versus the additional penalty coefficient $\beta$, where the initial inspection probability of MA of (a), (b), and (c) are all set as 0.5 while the initial sloth probability of MP is set as 0.3, 0.5, and 0.7.}
        \label{fig6}
    \end{center}
\end{figure}

By analyzing Figs. \ref{fig5} and \ref{fig6} from a holistic perspective, we can reveal that the \textbf{ideal situation}, i.e., both $\gamma$ and $\epsilon$ will be tiny as the game evolves, is not caused by continually raising the $\beta$. Instead, a modest $\beta$, $\beta = 11$, is conducive to facilitating the ideal situation. The basic rationale behind this phenomenon is two-fold. Firstly, from the deducing process of Proposition \ref{prop:1}, especially (\ref{eq:50}) and (\ref{eq:55}), we can identify that a larger $\beta$ is more likely render $\gamma = 0$, as shown in Fig. \ref{fig6}(a), thereby causing the property of $\gamma$ and $\epsilon$ consistent diminishing be invalid. Consequently, as depicted in Figs. \ref{fig5} (a) and (b), the $\epsilon$ is challenging to converge to $0$. Secondly, from the perspective of the utility model of MA, a larger $\beta$ implies that a substantial penalty will be imposed on the MP if MA detects MP is lazy while traffic data provision. To this end, the penalty will generate significant profit for the MA, causing MA to consistently supervise MP, as demonstrated in Fig. \ref{fig5}(c).

Upon the analysis of Figs. \ref{fig5} and \ref{fig6}, the experiment results in Figs. \ref{fig7} and \ref{fig8} is conducted under the setting of $\beta=11$. Observing Figs. \ref{fig7} and \ref{fig8}, we can identify that the impact of $\rho$ is less substantial than $\beta$ and the reason is analogous to the above analysis. In the following experiments, we set $\beta=11$ and $\rho=250$ as default values unless specified otherwise.

\begin{figure}[!t]
    \begin{center}
        \includegraphics[width=1\linewidth]{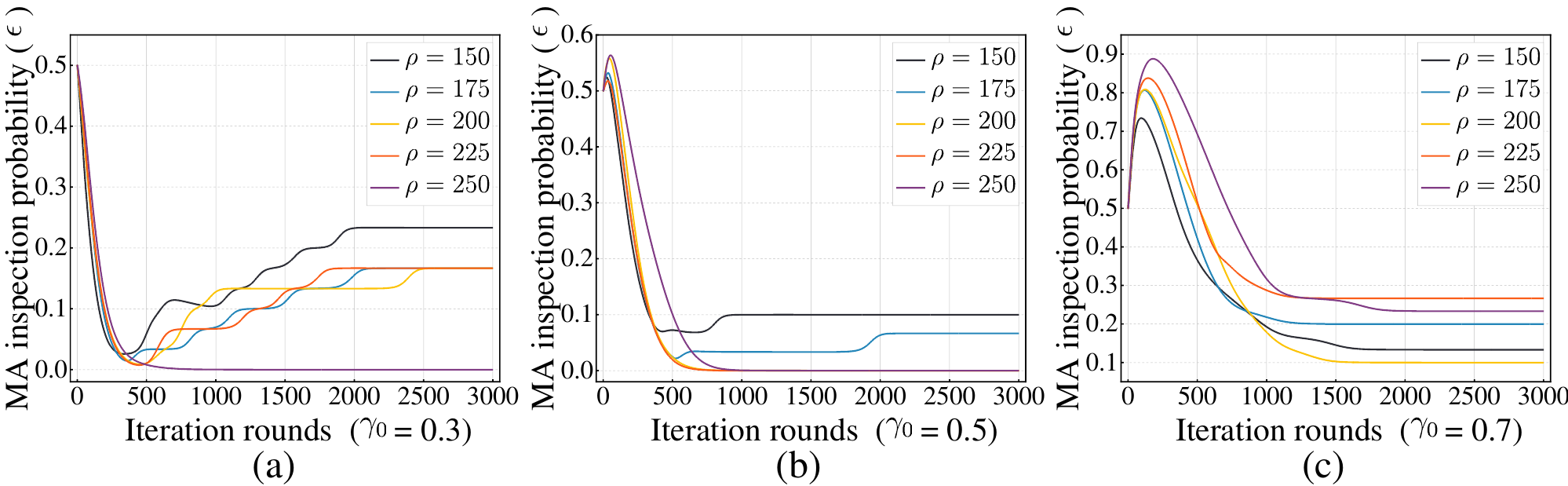}
        \caption{The inspection probability of MA $\epsilon$ versus the base penalty $\rho$, where the initial inspection probability of MA of (a), (b), and (c) are all set as 0.5 while the initial sloth probability of MP is set as 0.3, 0.5, and 0.7.}
        \label{fig7}
    \end{center}
\end{figure}

\begin{figure}[!t]
    \begin{center}
        \includegraphics[width=1\linewidth]{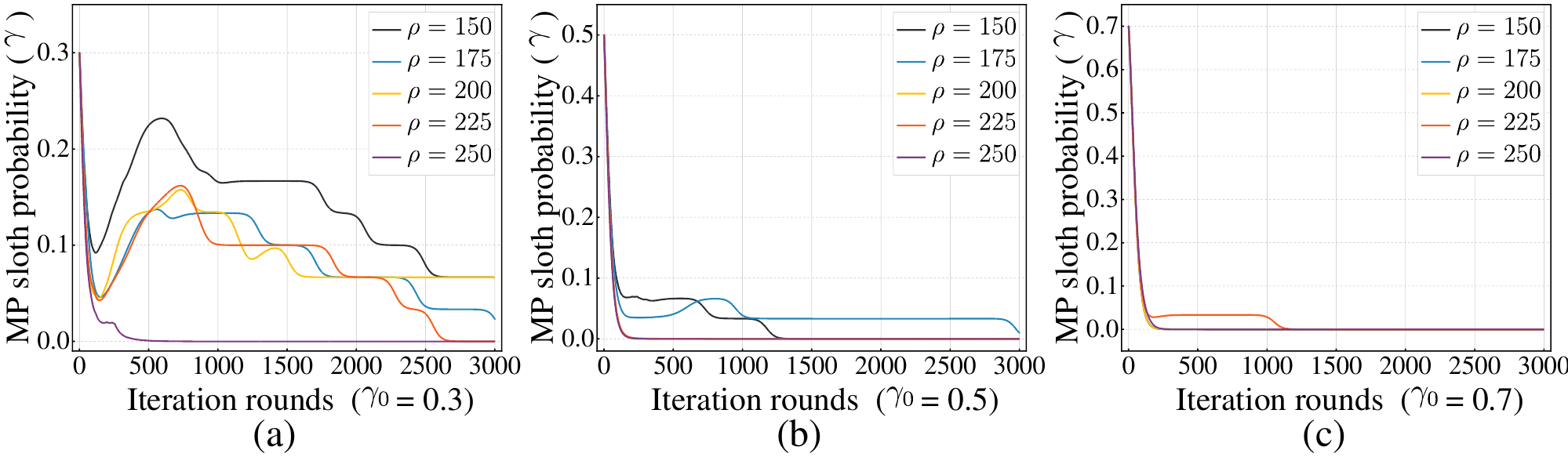}
        \caption{The sloth probability of MP $\gamma$ versus the base penalty $\rho$, where the initial inspection probability of MA of (a), (b), and (c) are all set as 0.5 while the initial sloth probability of MP is set as 0.3, 0.5, and 0.7.}
        \label{fig8}
    \end{center}
\end{figure}

\subsubsection{\textbf{Theoretical analysis}} \label{sec:6.3.3}
In this experiment, we discuss whether our proposed incentive mechanism is bounded by Propositions \ref{prop:1} and \ref{prop:2} and present the results in Fig. \ref{fig9}. Notably, the theoretical upper bound of $\gamma$ and $\epsilon$ primarily stem from (\ref{eq:50}), so we will assess the effectiveness of our proposed incentive mechanism via the probability gap $\gamma_{i^*}\epsilon_{i^*} - \gamma\epsilon$. Here, we consider that the theoretical upper bound $\gamma_{i^*}\epsilon_{i^*}$ is calculated as per the first inequality of (\ref{eq:50}), i.e., $\gamma_{i^*} \epsilon_{i^*} \rho (\beta - 1) = W$.

\begin{figure}[!t]
    \begin{center}
        \includegraphics[width=1\linewidth]{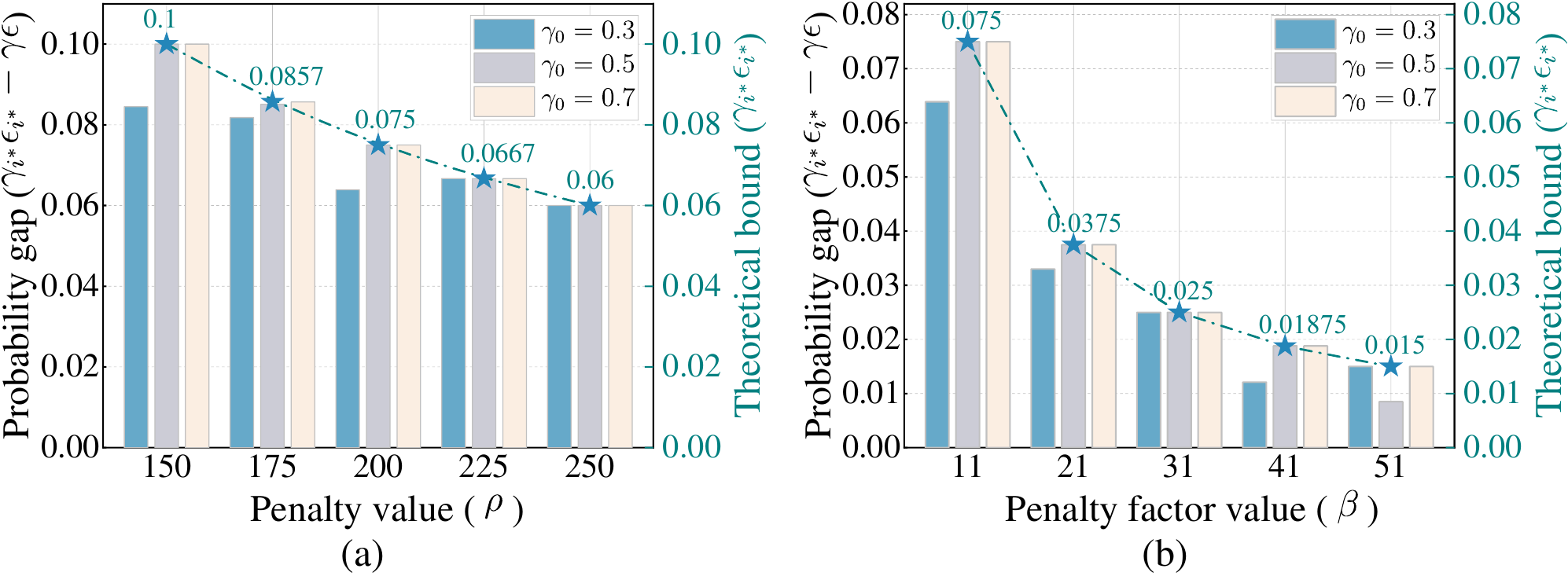}
        \caption{The probability gap between theoretical upper bound $\gamma_{i^*} \epsilon_{i^*}$ and converged sloth and inspection probability of MP and MA $\gamma \epsilon$ versus initial sloth probability of MA $\gamma_0$, where (a) and (b) varying the base penalty and additional penalty coefficient, respectively.}
        \label{fig9}
    \end{center}
\end{figure}

As shown in Figs. \ref{fig9} (a) and (b), the theoretical upper bound is diminishing along with the raising of base penalty $\rho$ and the additional penalty coefficient $\beta$, which is aligned with the Propositions \ref{prop:1} and \ref{prop:2}. In addition, regardless of the varying settings, the probability gap is consistently greater than $0$, indicating that the converged strategy of MA and MP is bounded. Furthermore, under certain settings, the probability gap is equal to the theoretical bound, which means that either $\gamma$ or $\epsilon$ has converged to $0$, reaching the second or third inequality of (\ref{eq:50}). Lastly, one interesting experiment result depicted in Figs. \ref{fig9} (a) and (b) is that almost exclusively only under the setting of $\gamma_0 = 0.3$ and $\rho$ and $\beta$ is in small value, the probability gap is not equivalent to the theoretical upper bound. The basic rationale is in line with the analysis presented in Section \ref{sec:6.3.2}, i.e., a larger $\gamma_0$, $\rho$, and $\beta$ cause an increased inspection probability of MA, thereby rendering $\gamma$ is more likely converge to $0$.

\subsubsection{\textbf{Comparison results}}
\begin{figure}[!t]
    \begin{center}
        \includegraphics[width=1\linewidth]{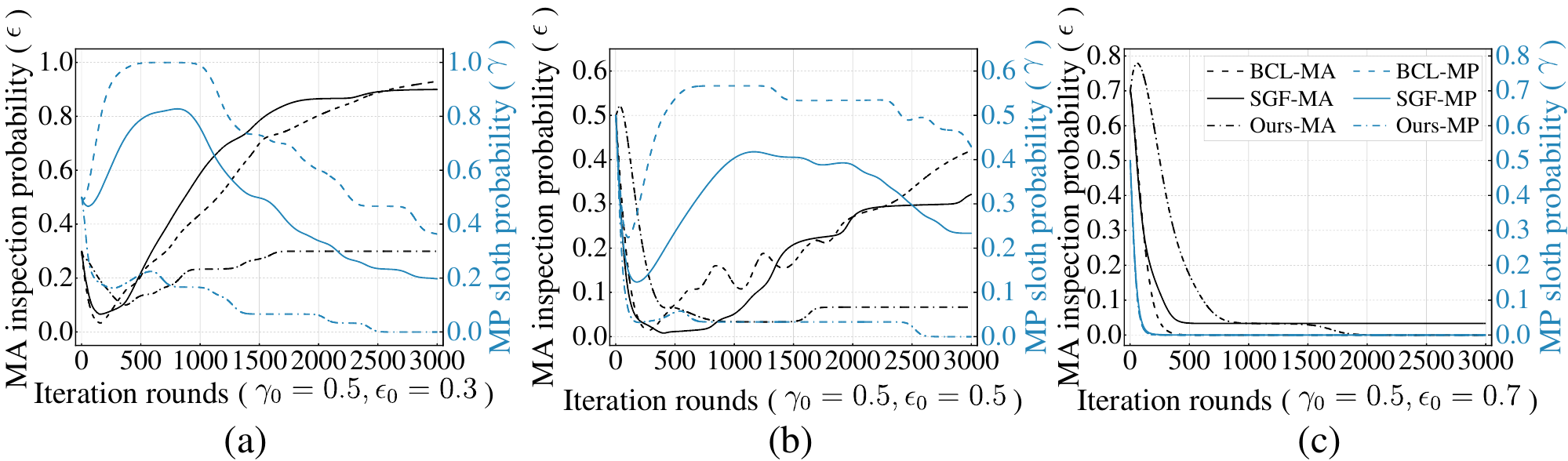}
        \caption{The comparison results in terms of MA's inspection probability $\epsilon$ and MP's sloth probability $\gamma$, where the initial sloth probability of MP of (a), (b), and (c) are all set as 0.5 while the initial inspection probability of MA is set as 0.3, 0.5, and 0.7.}
        \label{fig10}
    \end{center}
\end{figure}

\begin{figure}[!t]
    \begin{center}
        \includegraphics[width=1\linewidth]{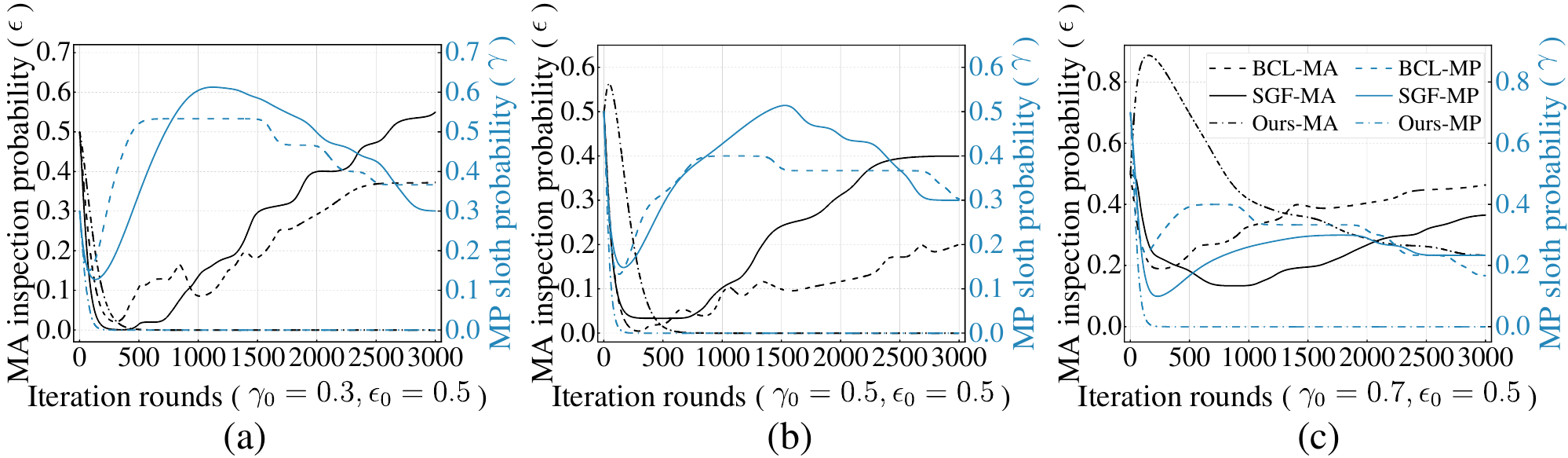}
        \caption{The comparison results in terms of MA's inspection probability $\epsilon$ and MP's sloth probability $\gamma$, where the initial inspection probability of MA of (a), (b), and (c) are all set as 0.5 while the initial sloth probability of MP is set as 0.3, 0.5, and 0.7.}
        \label{fig11}
    \end{center}
\end{figure}
In this section, we first compare the evolutionary curve of MA's inspection probability $\epsilon$ and MP's sloth probability $\gamma$ under varying settings with benchmarks in Figs. \ref{fig10} and \ref{fig11}. Subsequently, we compare the utility of MA, the utility of MP, and social welfare with benchmarks by varying $\epsilon_0$ and $\gamma_0$ in Figs. \ref{fig12} and \ref{fig13}.

As shown in Figs. \ref{fig10} and \ref{fig11}, regardless of the varying settings of $\gamma_0$ and $\epsilon_0$, our proposed incentive method can consistently ensure that converged sloth probability of MP $\gamma$ and inspection probability of MA $\epsilon$ approximate to $0$, excepts the results illustrated in Figs. \ref{fig10}(a) and \ref{fig11}(c). Although our proposed incentive mechanism cannot assure the inspection probability of MA $\epsilon$ converges to $0$ when $\gamma_0 > \epsilon_0$, it is still superior to benchmarks and is capable of converging the $\epsilon$ to a small value. Moreover, by comprehensively analyzing Figs. \ref{fig10} and \ref{fig11}, we can discern that our proposed incentive mechanism can diminish $\gamma$ and $\epsilon$ by around $23\%$ and $30\%$, respectively, compared to benchmarks.

\begin{figure}[!t]
    \begin{center}
        \includegraphics[width=1\linewidth]{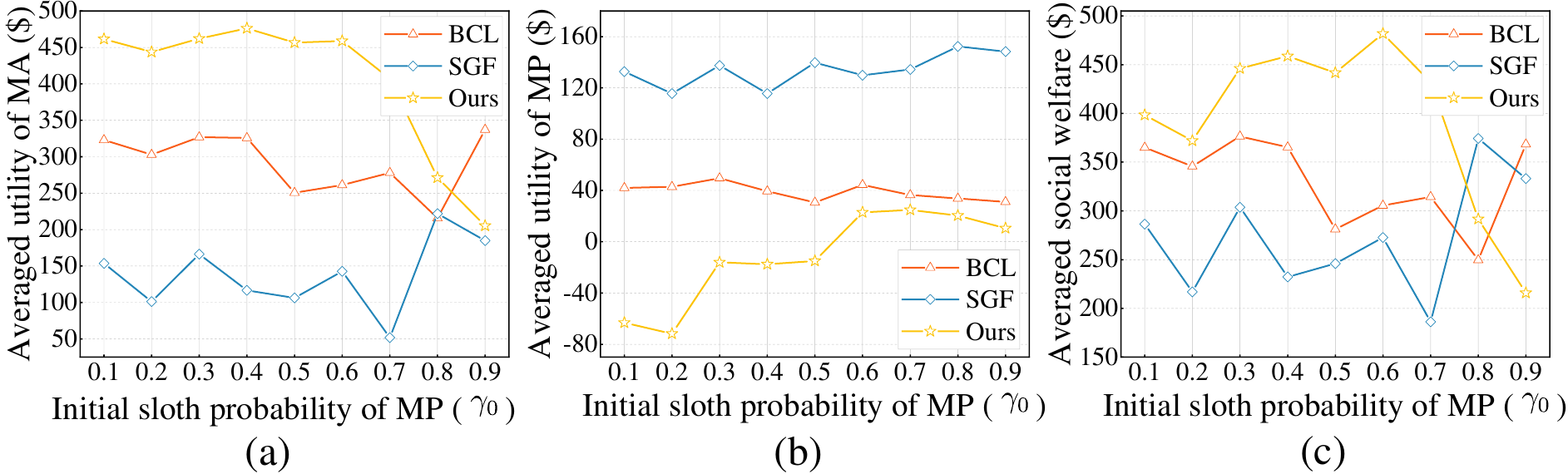}
        \caption{The comparison results with varied $\gamma_0$ and fixed $\epsilon_0 = 0.5$, where (a) is the averaged utility of MA, (b) is the averaged utility of MP, and (c) is the averaged social welfare.}
        \label{fig12}
    \end{center}
\end{figure}

\begin{figure}[!t]
    \begin{center}
        \includegraphics[width=1\linewidth]{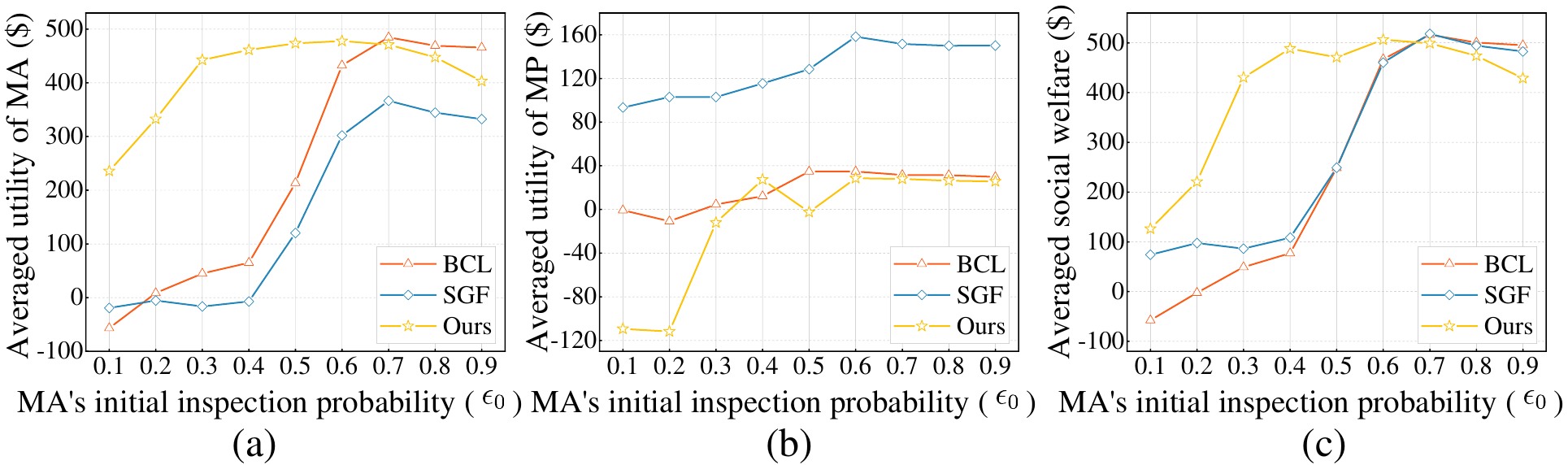}
        \caption{The comparison results with varied $\epsilon_0$ and fixed $\gamma_0 = 0.5$, where (a) is the averaged utility of MA, (b) is the averaged utility of MP, and (c) is the averaged social welfare.}
        \label{fig13}
    \end{center}
\end{figure}

By observing Fig. \ref{fig12}, we can find that the utility of MA and social welfare dramatically decreases and the utility of MP shifts from negative to positive when $\gamma_0 > 0.6$. The reason is that the substantial penalty imposed on the MP at the initial stage of the game spurs the MA to conduct supervision even when the MP's probability of sloth is $0$, resulting in unnecessary costs and reducing the utility of the MA. In parallel, due to MA's consistent inspection strategy, the sloth probability of MP swiftly decreases to $0$, so as to ensure the positive utility of MP, i.e., incentive rationale. Given that the supervision cost of MA significantly exceeds the profit of MP, rendering dramatically declined social welfare when $\gamma_0 > 0.6$. Notably, different from Fig. \ref{fig12}, the results presented in Fig. \ref{fig13} showcase that the utility of MA, the utility of MP, and social welfare is elevating when $\epsilon_0 < 0.4$.The basic rationale is that a small inspection probability facilitates MP to be lazy while traffic data provision, thereby undermining the utility of MA. Correspondingly, in light of the sloth probability of MP being elevated, the probability of MA detecting MP is lazy increases, resulting in the utility of MP declining.

As the comparison results presented in Figs. \ref{fig12}(b) and \ref{fig13}(b), the SGF is more tolerant to the lazy behavior of MP compared to our proposed method and BCL, considerably undermining the utility of MA, thereby rendering the least profit of MA. Additionally, by comprehensively analyzing Figs. \ref{fig12}(a) and \ref{fig13}(a), we can reveal that the utility of MA can be augmented by around $130\sim 400 \$$ in comparison with BCL.

\section{Conclusion} \label{sec:7}
In the context of VFL, we proposed an MI model-aided framework for reliable traffic state prediction, in which, the framework includes a data selection scheme for high-quality VFL model training and a penalty-based incentive mechanism for reliable VFL model running. Concretely, the data selection scheme first partitions the transportation network into multiple road segments. Next, the MA utilizes drone imaging to collect a small amount of traffic data to train an MI model for each road segment, where the MI model excels in capturing the relationship between traffic data and traffic states. Consequently, the MA can utilize the well-trained MI model to screen the most competent MP of each road segment to provide traffic data. Subsequently, the MA and screened MPs collectively train the VFL-based TSE model, which cannot only diminish the cost of MA but also safeguard the data privacy of MPs. Nonetheless, MP might be lazy in data provision during the VFL model running for augmented utility. To this end, we proposed an MI model-driven and penalty-based incentive mechanism to inhibit MPs' sloth probability. Concretely, the MA can leverage the well-trained MI model to inspect the data quality of MP, so as to detect whether MP is lazy in data provision. Upon this, we devised a supervision-based game model to incentivize MP and theoretically analyzed the feasibility of the incentive mechanism. The simulation results demonstrated that i) our proposed MI models can effectively identify the data quality, which in turn the data provider selection scheme can improve the traffic density and flow accuracy by 11.23\% and 21.15\% compared to the benchmark and ii) the incentive mechanism can converge MP's sloth probability and MA's inspection probability to $0$ under varying settings and augment the utility of MA by around $130\sim 400 \$$ compared to the benchmark.

\bibliographystyle{IEEEtran}
\bibliography{zhan}

\begin{thebibliography}{10}
\providecommand{\url}[1]{#1}
\csname url@samestyle\endcsname
\providecommand{\newblock}{\relax}
\providecommand{\bibinfo}[2]{#2}
\providecommand{\BIBentrySTDinterwordspacing}{\spaceskip=0pt\relax}
\providecommand{\BIBentryALTinterwordstretchfactor}{4}
\providecommand{\BIBentryALTinterwordspacing}{\spaceskip=\fontdimen2\font plus
\BIBentryALTinterwordstretchfactor\fontdimen3\font minus \fontdimen4\font\relax}
\providecommand{\BIBforeignlanguage}[2]{{%
\expandafter\ifx\csname l@#1\endcsname\relax
\typeout{** WARNING: IEEEtran.bst: No hyphenation pattern has been}%
\typeout{** loaded for the language `#1'. Using the pattern for}%
\typeout{** the default language instead.}%
\else
\language=\csname l@#1\endcsname
\fi
#2}}
\providecommand{\BIBdecl}{\relax}
\BIBdecl

\bibitem{pan2023ising}
Z.~Pan, A.~Sharma, J.~Y.-C. Hu, Z.~Liu, A.~Li, H.~Liu, M.~Huang, and T.~Geng, ``Ising-traffic: Using ising machine learning to predict traffic congestion under uncertainty,'' in \emph{Proceedings of the AAAI Conference on Artificial Intelligence}, Washington DC, Feb 2023.

\bibitem{yang2023traffic}
H.~Yang, L.~Du, G.~Zhang, and T.~Ma, ``A traffic flow dependency and dynamics based deep learning aided approach for network-wide traffic speed propagation prediction,'' \emph{Transportation Research Part B: Methodological}, vol. 167, pp. 99--117, Jan. 2023.

\bibitem{xu2024pigat}
Q.~Xu, Y.~Pang, X.~Zhou, and Y.~Liu, ``Pigat: Physics-informed graph attention transformer for air traffic state prediction,'' \emph{IEEE Transactions on Intelligent Transportation Systems}, vol.~25, no.~7, pp. 12\,561--12\,577, Sep. 2024.

\bibitem{wang2024privacy}
Q.~Wang and K.~Yang, ``Privacy-preserving data fusion for traffic state estimation: A vertical federated learning approach,'' \emph{Transportation Research Part C: Emerging Technologies}, vol. 168, p. 104743, Nov. 2024.

\bibitem{he2019optimal}
B.~Y. He and J.~Y. Chow, ``Optimal privacy control for transport network data sharing,'' \emph{Transportation Research Part C: Emerging Technologies}, vol. 113, no.~1, pp. 792--811, Apr. 2020.

\bibitem{ye2020federated}
D.~Ye, R.~Yu, M.~Pan, and Z.~Han, ``Federated learning in vehicular edge computing: A selective model aggregation approach,'' \emph{IEEE Access}, vol.~8, pp. 23\,920--23\,935, Jan. 2020.

\bibitem{liu2024vertical}
Y.~Liu, Y.~Kang, T.~Zou, Y.~Pu, Y.~He, X.~Ye, Y.~Ouyang, Y.-Q. Zhang, and Q.~Yang, ``Vertical federated learning: Concepts, advances, and challenges,'' \emph{IEEE Transactions on Knowledge and Data Engineering}, vol.~36, no.~7, pp. 3615--3634, Jul. 2024.

\bibitem{chougule2023novel}
A.~Chougule, V.~Chamola, V.~Hassija, P.~Gupta, and F.~R. Yu, ``A novel framework for traffic congestion management at intersections using federated learning and vertical partitioning,'' \emph{IEEE Transactions on Consumer Electronics}, vol.~70, no.~1, pp. 1725--1735, Feb. 2024.

\bibitem{errounda2022mobility}
F.~Z. Errounda and Y.~Liu, ``A mobility forecasting framework with vertical federated learning,'' in \emph{IEEE 46th Annual Computers, Software, and Applications Conference (COMPSAC)}, Torino, Italy, Jun 2022.

\bibitem{li2023fedsdg}
A.~Li, H.~Peng, L.~Zhang, J.~Huang, Q.~Guo, H.~Yu, and Y.~Liu, ``Fedsdg-fs: Efficient and secure feature selection for vertical federated learning,'' in \emph{IEEE Conference on Computer Communications}, New York, NY, May 2023.

\bibitem{jiang2022vf}
J.~Jiang, L.~Burkhalter, F.~Fu, B.~Ding, B.~Du, A.~Hithnawi, B.~Li, and C.~Zhang, ``Vf-ps: How to select important participants in vertical federated learning, efficiently and securely?'' in \emph{Advances in Neural Information Processing Systems}, New Orleans, LA, Nov 2022.

\bibitem{zhang2022secure}
R.~Zhang, H.~Li, M.~Hao, H.~Chen, and Y.~Zhang, ``Secure feature selection for vertical federated learning in ehealth systems,'' in \emph{IEEE International Conference on Communications}, Seoul, South Korea, May 2022.

\bibitem{feng2022vertical}
S.~Feng, ``Vertical federated learning-based feature selection with non-overlapping sample utilization,'' \emph{Expert Systems with Applications}, vol. 208, p. 118097, Dec. 2022.

\bibitem{castiglia2023less}
T.~Castiglia, Y.~Zhou, S.~Wang, S.~Kadhe, N.~Baracaldo, and S.~Patterson, ``Less-vfl: Communication-efficient feature selection for vertical federated learning,'' in \emph{International Conference on Machine Learning}, Hawaii, HI, Jul 2023.

\bibitem{arisdakessian2023coalitional}
S.~Arisdakessian, O.~A. Wahab, A.~Mourad, and H.~Otrok, ``Coalitional federated learning: Improving communication and training on non-iid data with selfish clients,'' \emph{IEEE Transactions on Services Computing}, vol.~16, no.~4, pp. 2462--2476, Jul-Aug. 2023.

\bibitem{liu2024chiron}
Y.~Liu, S.~Guo, Y.~Zhan, L.~Wu, Z.~Hong, and Q.~Zhou, ``Chiron: A robustness-aware incentive scheme for edge learning via hierarchical reinforcement learning,'' \emph{IEEE Transactions on Mobile Computing}, vol.~23, no.~8, pp. 8508--8524, Aug. 2024.

\bibitem{li2024rate}
B.~Li, J.~Lu, S.~Cao, L.~Hu, Q.~Dai, S.~Yang, and Z.~Ye, ``Rate: Game-theoretic design of sustainable incentive mechanism for federated learning,'' \emph{IEEE Internet of Things Journal}, Oct. 2024, early Access.

\bibitem{pan2024towards}
C.~Pan, J.~Xu, Y.~Yu, Z.~Yang, Q.~Wu, C.~Wang, L.~Chen, and Y.~Yang, ``Towards fair graph federated learning via incentive mechanisms,'' in \emph{Proceedings of the AAAI Conference on Artificial Intelligence}, Vancouver, Canada, Feb 2024.

\bibitem{zhao2024context}
Y.~Zhao, Y.~Qu, Y.~Xiang, F.~Chen, and L.~Gao, ``Context-aware consensus algorithm for blockchain-empowered federated learning,'' \emph{IEEE Transactions on Cloud Computing}, vol.~12, no.~2, pp. 491--503, Apr-Jun. 2024.

\bibitem{fu2023incentive}
Y.~Fu, C.~Li, F.~R. Yu, T.~H. Luan, and P.~Zhao, ``An incentive mechanism of incorporating supervision game for federated learning in autonomous driving,'' \emph{IEEE Transactions on Intelligent Transportation Systems}, vol.~24, no.~12, pp. 14\,800--14\,812, Dec. 2023.

\bibitem{kazmi2021novel}
S.~A. Kazmi, T.~N. Dang, I.~Yaqoob, A.~Manzoor, R.~Hussain, A.~Khan, C.~S. Hong, and K.~Salah, ``A novel contract theory-based incentive mechanism for cooperative task-offloading in electrical vehicular networks,'' \emph{IEEE Transactions on Intelligent Transportation Systems}, vol.~23, no.~7, pp. 8380--8395, Jul. 2022.

\bibitem{ning2020intelligent}
Z.~Ning, K.~Zhang, X.~Wang, L.~Guo, X.~Hu, J.~Huang, B.~Hu, and R.~Y. Kwok, ``Intelligent edge computing in internet of vehicles: A joint computation offloading and caching solution,'' \emph{IEEE Transactions on Intelligent Transportation Systems}, vol.~22, no.~4, pp. 2212--2225, Apr. 2021.

\bibitem{kim2017dual}
Y.~Kim, J.~Kwak, and S.~Chong, ``Dual-side optimization for cost-delay tradeoff in mobile edge computing,'' \emph{IEEE Transactions on Vehicular Technology}, vol.~67, no.~2, pp. 1765--1781, Feb. 2018.

\bibitem{lim2020hierarchical}
W.~Y.~B. Lim, Z.~Xiong, C.~Miao, D.~Niyato, Q.~Yang, C.~Leung, and H.~V. Poor, ``Hierarchical incentive mechanism design for federated machine learning in mobile networks,'' \emph{IEEE Internet of Things Journal}, vol.~7, no.~10, pp. 9575--9588, Oct. 2020.

\bibitem{belghazi2018mutual}
M.~I. Belghazi, A.~Baratin, S.~Rajeshwar, S.~Ozair, Y.~Bengio, A.~Courville, and D.~Hjelm, ``Mutual information neural estimation,'' in \emph{International conference on machine learning}, Stockholm, Sweden, Jul 2018.

\bibitem{zhaoL2024Verti}
F.~Zhao, Z.~Li, X.~Ren, B.~Ding, S.~Yang, and Y.~Li, ``Vertimrf: Differentially private vertical federated data synthesis,'' in \emph{Proceedings of the 30th {ACM} {SIGKDD} Conference on Knowledge Discovery and Data Mining}, Barcelona, Spain, Aug 2024.

\bibitem{yu2017spatio}
B.~Yu, H.~Yin, and Z.~Zhu, ``Spatio-temporal graph convolutional networks: A deep learning framework for traffic forecasting,'' in \emph{Proceedings of the Twenty-Seventh International Joint Conference on Artificial Intelligence}, Stockholm, Sweden, Jul 2018.

\bibitem{barmpounakis2020new}
E.~Barmpounakis and N.~Geroliminis, ``On the new era of urban traffic monitoring with massive drone data: The pneuma large-scale field experiment,'' \emph{Transportation Research Part C: Emerging Technologies}, vol. 111, pp. 50--71, Feb. 2020.

\bibitem{meert2018hmm}
W.~Meert and M.~Verbeke, ``Hmm with non-emitting states for map matching,'' in \emph{European Conference on Data Analysis (ECDA)}, Paderborn, Germany, Jul 2018.

\bibitem{fu2020vehicular}
Y.~Fu, F.~R. Yu, C.~Li, T.~H. Luan, and Y.~Zhang, ``Vehicular blockchain-based collective learning for connected and autonomous vehicles,'' \emph{IEEE wireless communications}, vol.~27, no.~2, pp. 197--203, Apr. 2020.

\end{thebibliography}

\begin{IEEEbiography}[{\includegraphics[width=1in,height=1.25in,clip,keepaspectratio]{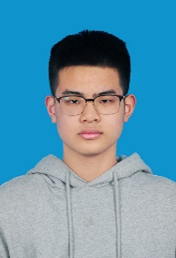}}]{Zijun Zhan} received his B.S. degree in software engineering from Northwest Normal University in 2021, and his M.S. degree in Computer Technology from China University of Petroleum (East China) in 2024. He is currently pursuing his Ph.D. degree in the Electrical and Computer Engineering Department at the University of Houston, Houston, TX, USA. His current research interests include Blockchain, Game Theory, and Intelligent Edge Computing.
\end{IEEEbiography}
\vspace{-1cm}
\begin{IEEEbiography}[{\includegraphics[width=1in,height=1.25in,clip,keepaspectratio]{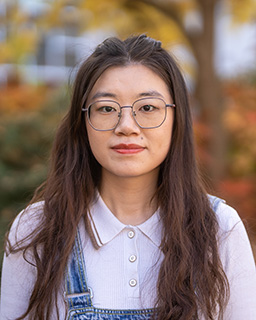}}]{Yaxian Dong} received her B.S. degree in quantity surveying from Chongqing University in 2019, and her M.Eng. degree in infrastructure system management from National University of Singapore in 2022. She is currently pursuing her Ph.D. degree in the department of architecture engineering at the Pennsylvania State University, University Park, PA, USA. Her current research interests include Building Information Modeling, Blockchain, and Agentic Workflow.
\end{IEEEbiography}
\vspace{-1cm}
\begin{IEEEbiography}[{\includegraphics[width=1in,height=1.25in,clip,keepaspectratio]{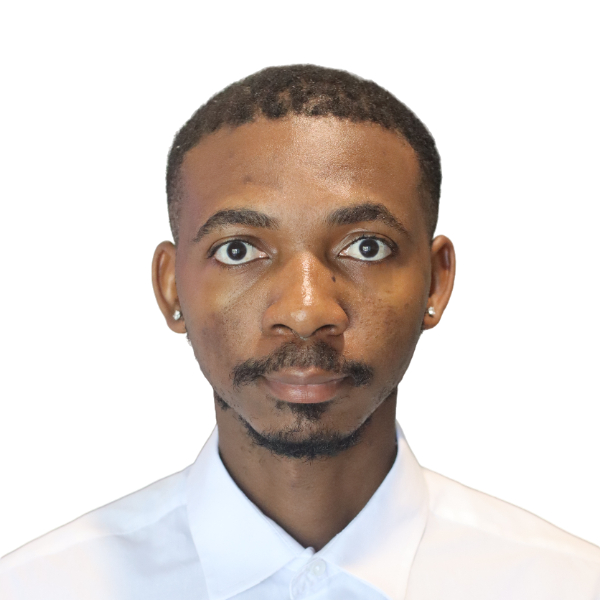}}]{Daniel Mawunyo Doe} received the bachelor’s degree in computer engineering from the Kwame Nkrumah University of Science and Technology, Kumasi, Ghana, in 2018. From 2019 to 2021, he pursued the M.Sc. degree in computer science and engineering at the University of Electronic Science and Technology of China (UESTC). He graduated with a Ph.D. in 2024 from the Electrical and Computer Engineering Department at the University of Houston, Houston, TX, USA. He is currently an assistant professor in the Department of Electrical and Computer Engineering at Prairie View A\&M University. His research interests include AI, blockchain, the Internet of Things, cryptography, and game theory.
\end{IEEEbiography}
\vspace{-1cm} 
\begin{IEEEbiography}[{\includegraphics[width=1in,height=1.25in,clip,keepaspectratio]{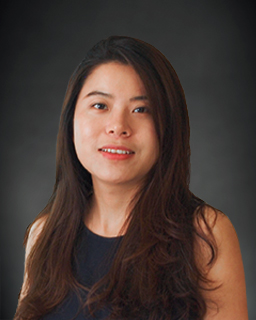}}]{Yuqing Hu} received the M.S. degree in management science and engineering from Tongji University, Shanghai, China, in 2016, and the M.S.degree in computational science and engineering and the Ph.D. degree in building construction from the Georgia Institute of Technology, Atlanta, GA, USA, in 2019 and 2020, respectively. She is currently an Assistant Professor of architectural engineering with Penn State University. Her research interests include building information modeling, building design and construction coordination, graph theory, and machine learning.
\end{IEEEbiography}
\vspace{-1cm} 
\begin{IEEEbiography}[{\includegraphics[width=1in,height=1.25in,clip,keepaspectratio]{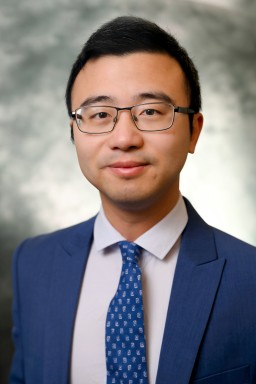}}]{Shuai Li} received his Ph.D. from Purdue University and previously served as an Assistant and Associate Professor at the University of Tennessee, Knoxville. He is currently an Associate Professor at the University of Florida. His research focuses on robotics and embodied artificial intelligence, with broad applications in construction, manufacturing, transportation, and healthcare. He has published over 60 journal papers along with numerous conference papers.
\end{IEEEbiography}
\vspace{-1cm} 
\begin{IEEEbiography}[{\includegraphics[width=1in,height=1.25in,clip,keepaspectratio]{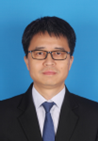}}]{Shaohua Cao} received the B.S. and master’s degrees from the China University of Petroleum (East China), Qingdao, China, in 2003 and 2013, respectively. He is currently an Associate Professor in the College of Computer Science and Technology at China University of Petroleum (East China). His research interests include Federated Learning, Edge Computing, Blockchain, Industrial Internet and SDN.
\end{IEEEbiography}
\vspace{-1cm} 
\begin{IEEEbiography}[{\includegraphics[width=1in,height=1.25in,clip,keepaspectratio]{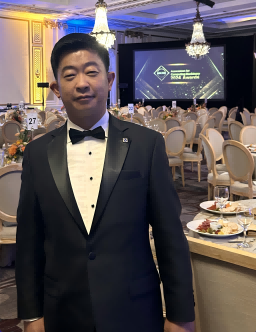}}]{Zhu Han} (S’01–M’04-SM’09-F’14) received the B.S. degree in electronic engineering from Tsinghua University, in 1997, and the M.S. and Ph.D. degrees in electrical and computer engineering from the University of Maryland, College Park, in 1999 and 2003, respectively. 

From 2000 to 2002, he was an R\&D Engineer of JDSU, Germantown, Maryland. From 2003 to 2006, he was a Research Associate at the University of Maryland. From 2006 to 2008, he was an assistant professor at Boise State University, Idaho. Currently, he is a John and Rebecca Moores Professor in the Electrical and Computer Engineering Department as well as in the Computer Science Department at the University of Houston, Texas. Dr. Han’s main research targets on the novel game-theory related concepts critical to enabling efficient and distributive use of wireless networks with limited resources. His other research interests include wireless resource allocation and management, wireless communications and networking, quantum computing, data science, smart grid, carbon neutralization, security and privacy.  Dr. Han received an NSF Career Award in 2010, the Fred W. Ellersick Prize of the IEEE Communication Society in 2011, the EURASIP Best Paper Award for the Journal on Advances in Signal Processing in 2015, IEEE Leonard G. Abraham Prize in the field of Communications Systems (best paper award in IEEE JSAC) in 2016, IEEE Vehicular Technology Society 2022 Best Land Transportation Paper Award, and several best paper awards in IEEE conferences. Dr. Han was an IEEE Communications Society Distinguished Lecturer from 2015 to 2018 and ACM Distinguished Speaker from 2022 to 2025, AAAS fellow since 2019, and ACM Fellow since 2024. Dr. Han is a 1\% highly cited researcher since 2017 according to Web of Science. Dr. Han is also the winner of the 2021 IEEE Kiyo Tomiyasu Award (an IEEE Field Award), for outstanding early to mid-career contributions to technologies holding the promise of innovative applications, with the following citation: ``for contributions to game theory and distributed management of autonomous communication networks."
\end{IEEEbiography}

\end{document}